\documentclass[11pt]{article}

\usepackage{enumerate}
\usepackage{mathtools, amssymb, latexsym, amsmath, amsfonts, amsthm, color}
\usepackage{array, fancyhdr, bm, listings}
\usepackage{algorithm}
\usepackage{algorithmic}
\usepackage{subcaption}
\mathtoolsset{showonlyrefs}

\theoremstyle{plain}
\newtheorem{theorem}{Theorem}[section]
\newtheorem{lemma}[theorem]{Lemma}
\newtheorem{proposition}[theorem]{Proposition}
\newtheorem{corollary}[theorem]{Corollary}

\theoremstyle{definition}
\newtheorem{definition}[theorem]{Definition}
\newtheorem{remark}[theorem]{Remark}

\DeclareMathOperator*\poly{poly}
\newcommand\defeq{\ensuremath{\stackrel{\rm def}{=}}} 
\DeclareMathOperator*{\argmin}{arg\,min}
\newcommand{\ind}[1]{^{(#1)}}
\newcommand\floor[1]{\lfloor#1\rfloor}
\newcommand\sar{^*}
\newcommand{\E}{\ensuremath{\mathbb{E}}} 

\newcommand{\ceil}[1]{\lceil {#1}\rceil}

\newcommand\citet{\cite}
\newcommand\citep{\cite}
\newcommand\SC{\textnormal{SC}}
\newcommand\MSSC{\textnormal{MSSC}}
\newcommand\WMSSC{\textnormal{WMSSC}}
\newcommand\bfp{\textbf{p}}

\newcommand\Kaariainen{K\"a\"ari\"ainen}
\newcommand\OPT{\textnormal{OPT}}
\newcommand\DT{\textnormal{DT}}

\begin{document}
\title{A Tight Analysis of Greedy Yields Subexponential Time Approximation for Uniform Decision Tree}
\author{Ray Li\thanks{Department of Computer Science, Stanford University.  Research supported by the National Science Foundation Graduate Research Fellowship Program under Grant No. DGE - 1656518. Email: \texttt{rayyli@cs.stanford.edu}} ,
Percy Liang\thanks{Department of Computer Science, Stanford University. Email: \texttt{pliang@cs.stanford.edu}} ,
Stephen Mussmann\thanks{Department of Computer Science, Stanford University. Research supported by the National Science Foundation Graduate Research Fellowship Program under Grant No. DGE - 1656518. Email: \texttt{mussmann@stanford.edu}}}
\date{\today}
\maketitle

\begin{abstract}
\textsc{Decision~Tree} is a classic formulation of active learning: given $n$ hypotheses with nonnegative weights summing to 1 and a set of tests that each partition the hypotheses, output a decision tree using the provided tests that uniquely identifies each hypothesis and has minimum (weighted) average depth. Previous works showed that the greedy algorithm achieves a $O(\log n)$ approximation ratio for this problem and it is NP-hard beat a $O(\log n)$ approximation, settling the complexity of the problem. 

However, for \textsc{Uniform~Decision~Tree}, i.e.\ \textsc{Decision~Tree} with uniform weights, the story is more subtle. The greedy algorithm's $O(\log n)$ approximation ratio was the best known, but the largest approximation ratio known to be NP-hard is $4-\varepsilon$. We prove that the greedy algorithm gives a $O(\frac{\log n}{\log C_{\OPT}})$ approximation for \textsc{Uniform~Decision~Tree}, where $C_{\OPT}$ is the cost of the optimal tree and show this is best possible for the greedy algorithm. As a corollary, we resolve a conjecture of Kosaraju, Przytycka, and Borgstrom \cite{kosaraju1999optimal}. Our results also hold for instances of \textsc{Decision~Tree} whose weights are not too far from uniform. Leveraging this result, for all $\alpha\in(0,1)$, we exhibit a $\frac{9.01}{\alpha}$ approximation algorithm to \textsc{Uniform~Decision~Tree} running in subexponential time $2^{\tilde O(n^\alpha)}$. As a corollary, achieving any super-constant approximation ratio on \textsc{Uniform~Decision~Tree} is not NP-hard, assuming the Exponential Time Hypothesis. This work therefore adds approximating \textsc{Uniform~Decision~Tree} to a small list of natural problems that have subexponential time algorithms but no known polynomial time algorithms. Like the analysis of the greedy algorithm, our analysis of the subexponential time algorithm gives similar approximation guarantees even for slightly nonuniform weights.
A key technical contribution of our work is showing a connection between greedy algorithms for \textsc{Uniform~Decision~Tree} and for \textsc{Min~Sum~Set~Cover}. 
\end{abstract}

\section{Introduction}

In \textsc{Decision~Tree} (also known as \textsc{Split Tree}), one is given $n$ hypotheses with nonnegative weights $p_1,\dots,p_n$ summing to 1 and a set of $m$ $K$-ary tests that each partition the hypotheses, and must output a decision tree using the provided tests that uniquely identifies each hypothesis and has minimum (weighted) average depth.\footnote{We require such a decision tree always exists in a valid \textsc{Decision~Tree instance}.}
\textsc{Decision~Tree} is a classic problem that arises naturally in active learning \citep{dasgupta2004analysis,nowak2011geometry,guillory2009average} and hypothesis identification \citep{moret1982decision}. 
Active learning with a well-specified and finite hypothesis class with noiseless tests is precisely \textsc{Decision~Tree} where the tests are data points and the answers are their labels.
\textsc{Decision~Tree} was first proved to be NP-hard by Hyafil and Rivest~\cite{HyafilR76}.
Since then, a large number works have provided algorithms for this question
\citep{GareyG74, Loveland85,kosaraju1999optimal,dasgupta2004analysis,chakaravarthy2007decision,chakaravarthy2009approximating,guillory2009average,gupta2010approximation,cicalese2010greedy, AdlerH12}.

A natural algorithm for \textsc{Decision~Tree} is the greedy algorithm, which creates a decision tree by iteratively choosing the test that most evenly splits the set of remaining hypotheses.
For binary tests ($K=2$), there is a natural notion of ``most even split,'' but for $K>2$, there are multiple possible definitions (see discussion in Section~\ref{sec:2}).
It is well known that the greedy algorithm achieves an $O(\log n)$ approximation ratio for \textsc{Decision~Tree} assuming all weights are at least $\frac{1}{\poly(n)}$. 
It was first shown for binary tests and uniform weights ($p_h = 1/n$ for all $h$) \citep{kosaraju1999optimal,AdlerH12}, then $K$-ary tests \citep{chakaravarthy2009approximating}, and finally, general non-uniform weights \citep{guillory2009average}. Furthermore, it is NP-hard to achieve a $o(\log n)$ approximation ratio for \textsc{Decision~Tree} \citep{chakaravarthy2007decision}, settling the complexity of approximating \textsc{Decision~Tree}.

However, there are still gaps in our knowledge.
For \textsc{Uniform~Decision~Tree}, i.e.\ \textsc{Decision~Tree} with uniform weights, the $O(\log n)$ approximation given by the greedy algorithm was previously the best known approximation achievable in polynomial time.
Chakaravarthy et al. \citep{chakaravarthy2007decision} proved that it is NP-hard to give a $(4-\varepsilon)$ approximation, giving the best known hardness of approximation result, and they asked whether the gap between the best approximation and hardness results could be improved.
Previously, it was not even known whether the greedy algorithm could beat the $O(\log n)$ approximation ratio in previous analyses: the best lower bound on the greedy algorithm's approximation ratio is $\Omega(\frac{\log n}{\log \log n})$ \citep{kosaraju1999optimal,dasgupta2004analysis}. 
In the setting where the optimal solution to \textsc{Uniform~Decision~Tree} has cost $O(\log n)$, Kosaraju et al. \cite{kosaraju1999optimal} showed that the greedy algorithm indeed gives an $O(\frac{\log n}{\log\log n})$ approximation,
and they conjectured that the greedy algorithm gives an $O(\frac{\log n}{\log\log n})$ approximation in general.

For an extended discussion of related works, see Section~\ref{sec:rel}.
\subsection{Our contributions} 
We summarize the main contributions of our work below. The approximation guarantees of our algorithms are captured in Figure~\ref{fig:results}.
\begin{itemize}
\item 
\textbf{Greedy algorithm.}
We give a new analysis of the greedy algorithm, showing that it gives an $O(\frac{\log n}{\log C_{\OPT}} + \log\frac{p_{\max}}{p_{\min}})$ approximation for \textsc{Decision~Tree}, where $C_{\OPT}$ is the cost of the optimal tree, $p_{\max}=\max_h p_h$, and $p_{\min}=\min_hp_h$.
This implies an $O(\frac{\log n}{\log C_{\OPT}})$ approximation for instances of \textsc{Uniform~Decision~Tree} and of \textsc{Decision~Tree} whose weights are close to uniform.
As $C_{\OPT}\ge\log_K n$ always, this proves the conjecture of Kosaraju et al. \cite{kosaraju1999optimal}.

\item
\textbf{Subexponential time algorithm.}
Leveraging the above greedy analysis, for $\alpha<1$, we give a subexponential\footnote{Throughout this work, subexponential means $2^{n^\alpha}$ for some absolute $\alpha\in(0,1)$. We make a distinction when referring to $2^{n^{o(1)}}$ runtimes.} $2^{\tilde O(n^{\alpha})}$-time $\frac{9.01}{\alpha}$ approximation algorithm for \textsc{Uniform~Decision~Tree}. 
Assuming the Exponential Time Hypothesis (ETH) \citep{ImpagliazzoP01, ImpagliazzoPZ01}\footnote{ETH states that there are no $2^{n^{o(1)}}$ time algorithms for 3SAT.}, this algorithm implies that any superconstant approximation of \textsc{Uniform~Decision~Tree} is \emph{not} NP-hard.
Our work adds approximating \textsc{Uniform~Decision~Tree} to a small list of natural problems whose time complexity is known to be subexponential (and, for some approximation ratios, $2^{n^{o(1)}}$) but not known to be polynomial.
Examples of such problems include \textsc{Factoring} \citep{lenstra1993number}, \textsc{Unique~Games} \citep{Khot02a,AroraBS15}, \textsc{Graph~Isomorphism} \citep{Babai16}, and approximating \textsc{Nash Equilibrium} \citep{LiptonMM03,Rubinstein18}, with the later two having $2^{n^{o(1)}}$-time algorithms.
Like in our analysis of the greedy algorithm, our subexponential time algorithm gives a similar approximation guarantee even for slightly nonuniform weights, in particular when $\frac{p_{\max}}{p_{\min}}\le 2^{O(1/\alpha)}$.

\item \textbf{Approximation ratio tightness.}
We prove that the $O(\frac{\log n}{\log C_{\OPT}})$ approximation ratio for the greedy algorithm is tight for \textsc{Uniform~Decision~Tree}.
We also prove that the $O(\log \frac{p_{\max}}{p_{\min}})$ term in the approximation ratio for the greedy algorithm is necessary, in the sense that no algorithm can give a $o(\log \frac{p_{\max}}{p_{\min}})$ approximation for \textsc{Decision~Tree} when $\frac{p_{\max}}{p_{\min}} = n^r$ for some $r\in(0,1)$ unless P=NP.

\item
\textbf{Repeatable, noisy tests.} \Kaariainen \cite{kaariainen2006active} provides a method to convert a solution for \textsc{Decision Tree} into a solution for a variant of \textsc{Decision~Tree} that handles noisy, repeatable tests. An immediate corollary of our result for the greedy algorithm is that the cost of a solution for the noisy problem derived from the greedy algorithm is at most $C_{\OPT}\cdot O(\log n \log \log n)$. 
Previously, this cost was bounded by $C_{\OPT}\cdot O(\log^2 n \log \log n)$. 

\end{itemize}

\subsection{Techniques}
Our work gives a new analysis of the greedy algorithm for \textsc{Decision~Tree}.
A key technical contribution of this work is to leverage upper bounds of \textsc{Min~Sum~Set~Cover} and \textsc{Set~Cover} for \textsc{(Uniform)~Decision~Tree}.
Previously, only connections in the reverse direction (i.e. lower bounds) were known between these problems: NP-hardness of attaining a $4-\varepsilon$-approximation for \textsc{Uniform~Decision~Tree} was proved by reduction from \textsc{Min~Sum~Set~Cover}, and NP-hardness of attaining a $o(\log n)$ approximation for \textsc{Decision~Tree} was proved by reduction from \textsc{Set~Cover} \cite{chakaravarthy2007decision}.

At a high level, our analysis goes as follows.
By a simple double counting argument, we can compute the cost of a tree by summing the ``weights'' of the tree's interior vertices, rather than summing the depths of the hypotheses.
However, rather than accounting for all the interior vertices at once, we separately analyze the vertices with ``imbalanced'' splits and those with ``balanced'' splits.
Carefully choosing the definition of balanced and imbalanced is a key idea of the proof: previous analyses of the greedy algorithm \cite{kosaraju1999optimal, chakaravarthy2007decision, guillory2009average, AdlerH12} either make no distinction between interior vertices or use a different distinction.
A global entropy argument accounts for the vertices with balanced splits.
For the vertices with imbalanced splits, we use the fact that the greedy algorithm gives a constant factor approximation for \textsc{Min~Sum~Set~Cover} \citep{FeigeLT04}.
For \textsc{Uniform~Decision~Tree}, putting the two bounds together gives the desired approximation result.
For the general \textsc{Decision~Tree} problem, we additionally prove and use a generalization of a result on the greedy algorithm's performance for \textsc{Set~Cover} \citep{lovasz1975ratio, johnson1974approximation,chvatal1979greedy, Stein74}.

For the subexponential time algorithm, we leverage our new result that the greedy algorithm gives an $O(\frac{\log n}{\log C_{\OPT}})$ approximation.
We first run the greedy algorithm. If the greedy algorithm returns a tree with cost at least $n^{3\alpha/4}$, we return the greedy tree knowing we have an $O(1/\alpha)$ approximation. Otherwise, we find by brute force the ``optimal tree up to depth $n^{\alpha}$'' in time $2^{\tilde O(n^{\alpha})}$, then recurse.

\begin{figure}
  \begin{center}
    \includegraphics[width=0.7\textwidth]{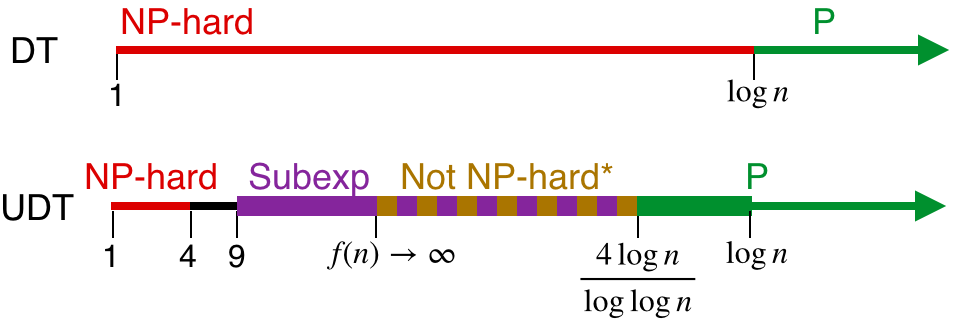}
  \end{center}
  \caption{Approximation ratio differences between \textsc{Decision~Tree} (DT) and \textsc{Uniform~Decision~Tree} (UDT). The thicker lines represent the contributions of this work. The yellow approximation region is not NP-hard assuming the Exponential Time Hypothesis (ETH).}
  \label{fig:results}
\end{figure}

\subsection{Organization of paper}

In Section~\ref{sec:2}, we formally introduce notation used throughout the paper.
In Section~\ref{sec:3}, we state our results.
In Section~\ref{sec:main-sketch}, we sketch a proof of Theorem~\ref{thm:alg-main}, that the greedy algorithm gives an $O(\frac{\log n}{\log C_{\OPT}} + \log \frac{p_{\max}}{p_{\min}})$ approximation on \textsc{Decision~Tree}.
Since the proof of Theorem~\ref{thm:alg-main} is involved, we prove the special case of Theorem~\ref{thm:alg-main} for \textsc{Uniform~Decision~Tree} with binary tests in Section~\ref{sec:main-proof}, and give the full proof in Appendix~\ref{app:Z}.
In Section~\ref{sec:subexp}, we state the subexponential time approximation algorithm and give a sketch of the analysis, and we give a formal analysis in Section~\ref{app:C}.
In Section~\ref{sec:rel}, we describe some related work.
In Section~\ref{sec:conclusion}, we conclude with some open problems.

We leave some details to the appendices.
A lemma on the greedy algorithm's performance in a generalization of \textsc{Set~Cover} that is used in the proof Theorem~\ref{thm:alg-main} is proved in Appendix~\ref{app:sc}.
In Appendix~\ref{app:tight}, we prove Propositions~\ref{thm:avg-lower} and \ref{thm:lower_bound_poly_ratio}, which show two ways that Theorem~\ref{thm:alg-main} is tight.
In Appendix~\ref{app:H}, we demonstrate a rounding trick that allows us to assume $p_{\min}\ge \frac{1}{\poly(n)}$ without changing the difficulty of approximating \textsc{Decision~Tree}.


\section{Preliminaries}
\label{sec:2}

For a positive integer $a$, let $[a]=\{1,\dots,a\}$.
All logs are base 2 when the base is not specified.
The \textsc{Decision~Tree} problem is as follows: given a set of hypotheses $[n]$ with probabilities $p_1,\dots,p_n$ summing to 1, and $m$ distinct $K$-ary tests, output a \emph{decision tree} $\mathcal{T}$ with hypotheses as leaves, such that the weighted average of the depth of the leaves is minimal. 
Formally, a $K$-ary test is a map $\tau:[n]\to [K]$.
We refer to $K$ as the \emph{branching factor} of the test $\tau$, and the elements of $[K]$ as the possible \emph{answers} to the tests.
We think of a test $\tau$ as defining a $K$-way partition of $[n]$.
A \emph{decision tree} $\mathcal{T}$ is a rooted tree such that each interior vertex $v$ has the index $j_v \in [m]$ of some test, and the edge to the $i$-th child of $v$ is labeled with $i \in [K]$.
We say that a hypothesis $h\in[n]$ is \emph{consistent} with a vertex $v$ if, in the root-to-$v$ path, the edge following any vertex $u$ has label $\tau_{j_u}(h)$.
We let $L(v)$ denote the set of hypotheses $h$ that are consistent with $v$.
We say a decision tree $\mathcal{T}$ is \emph{complete} if, for all $h\in[n]$, there exists a (unique) leaf $v\in\mathcal{T}$ such that $L(v)=\{h\}$, and for a complete decision tree $\mathcal{T}$, let $d_{\mathcal{T}}(h)$ denote the depth of this vertex $v$.
The \emph{cost} of a complete decision tree $\mathcal{T}$ is defined to be the average depth of the leaves, weighted by $p_h$, i.e.\
\begin{align}
  C(\mathcal{T}) \ &\defeq \  \sum_{h\in [n]}^{} p_h d_\mathcal{T}(h).
\label{}
\end{align}
We set $\mathcal{T}_\OPT$ to be a complete decision tree that minimizes $C(\mathcal{T}_{\OPT})$ (in general, there may be more than one optimal decision tree), and abbreviate $C_{\OPT}\defeq C(\mathcal{T}_\OPT)$.

This paper is concerned with the \emph{greedy algorithm} for \textsc{Decision~Tree}.
We call a decision tree \emph{greedy} if the test $\tau_{j_v}$ of each interior vertex $v$ minimizes the (weighted) number of hypotheses of the largest partition in $\tau_{j_v}$'s partitioning of $L(v)$.
Formally, a decision tree $\mathcal{T}$ is \emph{greedy} if, for all interior vertices $v\in\mathcal{T}$, we have
\begin{align}
  j_v \in \argmin_{j} \left( \max_{k\in[K]}p(L(v)\cap\tau_j^{-1}(k)) \right),
\label{}
\end{align}
where $p(S) = \sum_{h\in S}^{} p_h$ for $S\subseteq[n]$.
Given a \textsc{Uniform~Decision~Tree} instance, we let $\mathcal{T}_{G}$ be a complete, greedy decision tree, choosing one arbitrarily if there is more than one.
For brevity, we write $C_G \defeq C(\mathcal{T}_G)$.

We remark that, when $K>2$, our notion of a ``greedy'' algorithm for \textsc{Decision~Tree} is not the only one.
As mentioned in the previous paragraph, our definition of greedy chooses, at each vertex in the decision tree, the test that minimizes the (weighted) number of candidate hypotheses, assuming a worst-case answer to the test. Our definition corresponds to the definition by \cite{chakaravarthy2009approximating}, but other choices include maximizing the (weighted) number of pairs of hypotheses that are distinguished \citep{chakaravarthy2007decision, guillory2009average} and maximizing the mutual information between the test and the remaining hypotheses \citep{zheng2005efficient}. 
For binary tests, $K=2$, these definitions are all equivalent.

Define $\DT(R)$ as \textsc{Decision~Tree} with the guarantee that $\frac{p_\text{max}}{p_\text{min}} \leq R$. 
In this notation, $\DT(1)$ is \textsc{Uniform~Decision~Tree}.

\section{Our results}
\label{sec:3}

\subsection{Greedy algorithm}

The main driver of this paper is Theorem~\ref{thm:alg-main}, which relates the cost of the greedy algorithm to the optimal cost for \textsc{Decision~Tree}.

\begin{theorem}
  \label{thm:alg-main}
  For any instance of \textsc{Decision~Tree} on $n$ hypotheses, we have
  \begin{align}
    \label{eq:main}
    C_G \le \left(\frac{12\cdot\log(\frac{1}{p_\text{min}})}{\log C_{\OPT}} + \ln\left( \frac{p_\text{max}}{p_\text{min}} \right) \right)\cdot C_{\OPT}.
  \end{align}
\end{theorem}
Our theorem holds for any branching factor $K$, and when $C_G$ is the cost of an arbitrary tree produced by the greedy algorithm above.
As $C_{\OPT}\ge \log_K n$ always, our result implies that the greedy algorithm always gives an $O(\frac{\log n}{\log\log n})$ approximation for \textsc{Uniform~Decision~Tree} when the branching factor is a constant, resolving the conjecture of \cite{kosaraju1999optimal}.
Additionally, if $C_{\OPT}$ is $\Omega(n^\alpha)$ for constant $\alpha$ and the weights are uniform, then the greedy algorithm obtains a constant $O(1/\alpha)$ approximation.
We use this fact crucially in designing our subexponential time approximation algorithms.

For the simpler case when $K=2$ and the weights $p_h$ are uniform, we give a sketch of the proof in Section~\ref{sec:main-sketch} and a full proof in Section~\ref{sec:main-proof}.
This full result is sketched in Section~\ref{sec:main-sketch} and proven in full in Appendix~\ref{app:Z}.
For \textsc{Uniform~Decision~Tree}, the constant 12 can be improved to 6, and, when $n$ is sufficiently large, $4+\varepsilon$, so that greedy gives a $\frac{(4+\varepsilon)\log n}{\log C_{\OPT}}$ approximation (see Section~\ref{sec:main-sketch}).

Note that the terms $\log(\frac{1}{p_\text{min}})$ and $\ln(\frac{p_\text{max}}{p_\text{min}})$ in the approximation ratio can be arbitrarily large.
However, a rounding trick before running the greedy algorithm \citep{guillory2009average} allows us assume that all the weights are at least $\frac{1}{n^2}$, and hence $\log(\frac{1}{p_\text{min}})\le 2\log n$ and $\ln(\frac{p_\text{max}}{p_\text{min}})\le 2\ln n$ in \eqref{eq:main}.
The details are given in Appendix~\ref{app:H}.

\subsection{Subexponential time algorithm}
\label{sec:subexp_alg}

Using Theorem~\ref{thm:alg-main}, we give a subexponential time algorithm that achieves a constant factor approximation for the \textsc{Decision~Tree} problem when the weights are close to uniform.
\begin{theorem}
\label{thm:subexp-alg}
For any $\alpha\in(0,1)$ and $R\ge 1$, there exists an $(\frac{25}{\alpha}+\log R)$ approximation algorithm for $\DT(R)$ with runtime $2^{O( n^{\alpha} \log(Rn) \log m)}$.
For \textsc{Uniform~Decision~Tree}, for any $\varepsilon>0$, we can achieve a $\frac{9+\varepsilon}{\alpha}$ approximation in the same runtime.
\end{theorem}
In Section~\ref{sec:subexp}, the subexponential time algorithm is stated and an analysis is sketched.
The analysis is given formally in Section~\ref{app:C}.
Importantly, this result implies that achieving a super-constant approximation ratio is not NP-hard, given the Exponential Time Hypothesis. 
As an informal proof, suppose for contradiction there was a polynomial reduction from 3-SAT to achieving a $f(n)$ approximation ratio for \textsc{Uniform~Decision~Tree} for some $f(n)\to\infty$ as $n\to\infty$. 
By Theorem~\ref{thm:subexp-alg}, there exists a $2^{n^{o(1)}}$-time algorithm to achieve a $f(n)$ approximation for \textsc{Uniform~Decision~Tree}, and thus a $2^{n^{o(1)}}$-time algorithm to solve 3-SAT, contradicting the Exponential Time Hypothesis.
This adds approximating \textsc{Uniform~Decision~Tree} to a list of interesting natural problems that have subexponential or $2^{n^{o(1)}}$ time algorithms but are not known to be in P.
Figure~\ref{fig:results} illustrates the contrast between \textsc{Decision~Tree} and \textsc{Uniform~Decision~Tree}.

\subsection{Approximation ratio tightness}
We also show that the $O(\frac{\log n}{\log C_{\OPT}})$ approximation ratio is tight up to a constant factor for the greedy algorithm by generalizing the example given by \cite{dasgupta2004analysis}. 
The proof is given in Appendix~\ref{app:B}.
\begin{proposition}
\label{thm:avg-lower}
There exists an $n_0$ such that for all $n\ge n_0$ and any $C^*\in[\log n,n]$, there exists an instance of \textsc{Uniform~Decision~Tree} with branching factor 2 for which
\begin{align}
  C_G\ge \frac{C^* \log n}{16\log C^*}, \qquad\text{and}\qquad
  C_{\OPT} \le 8C^*.
\end{align}

\end{proposition}
We also show that, when the weights are non-uniform, the $\ln\left(\frac{p_\text{max}}{p_\text{min}}\right)$ term in the approximation ratio of Theorem~\ref{thm:alg-main} is computationally necessary.
\begin{proposition}
\label{thm:lower_bound_poly_ratio}
Let $r\in(0,1)$. Then, for $n$ sufficiently large, approximating $\DT(2n^r \log n)$ to a factor of $\frac{r}{12}\log n$ is NP-hard.
\end{proposition}
In other words, even if the ratio $p_\text{max}/p_\text{min}$ is guaranteed to be $O(n^r)$ for a constant $r\in(0,1)$, one cannot give a $o(\log n)$ approximation algorithm unless $\text{P}=\text{NP}$. 
The proof is given in Appendix~\ref{app:E}.

\subsection{Decision tree with noise}
\label{sec:noise}
Theorem~\ref{thm:alg-main} implies an improved black-box result for a noisy variant of \textsc{Decision~Tree}.
{\Kaariainen} \cite{kaariainen2006active} considers a variant of \textsc{Decision~Tree} with binary tests where the output of each test may be corrupted by i.i.d. noise. 
Formally, there exists $\varepsilon>0$ such that querying any test $\tau$ on any hypothesis $h$, outputs the correct answer $\tau(h)$ with probability $1-\delta_{\tau,h}$ and the wrong answer with probability $\delta_{\tau,h}$, for some $\delta_{\tau,h}\in[0,1/2-\varepsilon)$.
Tests are repeatable, with each one producing different draws of the noise.
\Kaariainen~\cite{kaariainen2006active} gives an algorithm that turns a decision tree of cost $C$ for the noiseless problem into a decision tree with cost $O(C\log C\log\log C)$ for the noisy problem by repeating queries sufficiently many times.

Combining \Kaariainen's result with the greedy algorithm for \textsc{Uniform~Decision~Tree} gives an algorithm for the noisy problem using an average of $O(C_G\log C_G \log\log C_G)$ queries.
Previously, using the bound $C_G\le \max(C_{\OPT}\cdot O(\log n), n)$, the noisy problem's cost was bounded by $C_{\OPT}\cdot O(\log^2n\log\log n)$.
However, by Theorem~\ref{thm:alg-main}, we have $C_G\le C_{\OPT}\cdot O(\frac{\log n}{\log C_{\OPT}})$, so we in fact have cost at most $C_{\OPT}\cdot O(\log n\log\log n)$, improving the cost ratio to the optimal solution of the noiseless problem by a nearly quadratic factor.

\section{Sketch of proof of Theorem~\ref{thm:alg-main}}
\label{sec:main-sketch}

In this section, we sketch a proof of Theorem~\ref{thm:alg-main}.
We first sketch the proof assuming that the branching factor is 2, so that $\mathcal{T}_G$ is a binary tree, and that the distribution is uniform ($p_h=1/n$ for all $h\in[n]$).
Since the proof of Theorem~\ref{thm:alg-main} is involved, we give the details of this easier result in Section~\ref{sec:main-proof}.
At the end of the section, we give the additional ideas necessary to complete the full proof of Theorem~\ref{thm:alg-main}.
The details of the full proof are given in Appendix~\ref{app:Z}.

\subsection{Uniform weights and binary tests}

Recall that $p(v)\defeq \sum_{h\in L(v)}^{} p_h$ and that, as the weights are uniform, $p_h=\frac{1}{n}$ for all $h\in[n]$.
By a simple double counting argument (Lemma~\ref{lem:alg-10}), we can compute the cost of the greedy tree by summing the weights of the vertices rather than summing the depths of leaves.
That is,
\begin{align}
C_G
\ = \  \sum_{v\in\mathcal{T}_G^{\mathrm{o}}}^{} p(v),
\end{align}
where the sum is over the interior vertices of $\mathcal{T}_G$.

\paragraph{Defining balanced and imbalanced vertices.}
We then define \emph{balanced} and \emph{imbalanced} vertices with respect to a parameter $\delta\in(0,1)$, which we eventually set to $O(\frac{1}{C_{\OPT}})$.
These definitions are crucial to the proof.
A vertex is \emph{imbalanced}\footnote{We remark that imbalanced vertices can have $p(v^-)$ arbitrarily close to $\frac{p(v)}{2}$, so the hypotheses at vertex $v$ are not necessary split in an imbalanced way. However, as we show (Lemma~\ref{lem:type-2}), all balanced vertices are in fact split in a balanced way with $p(v^-)/p(v)\ge\delta/2$, hence the terminology.} if there exists an \emph{integer} $s$ (called the \emph{level}) such that $p(v) > 2\delta^s$ and $p(v^-)\le \delta^s$. 
Here, $v^-$ is the child of $v$ containing a smaller weight of hypotheses in its subtree. We say $v$ is \emph{balanced} if it is not imbalanced.
Note that imbalanced vertices exist only for $s\le s_\text{max}$, where $s_\text{max}\defeq \frac{\log n}{\log\frac{1}{\delta}}$.
We prove a structural result (Lemma~\ref{lem:type-1}) that shows that the level-$s$ imbalanced vertices of $\mathcal{T}_G$ can be partitioned into downward paths, which we call \emph{chains}, such that, for all $s$, each leaf has vertices from at most one level-$s$ chain among its ancestors.
The parameter $\delta$ quantifies how many chains we consider: smaller $\delta$ means fewer, longer chains, and larger $\delta$ means more, shorter chains.
We optimize the choice of $\delta$ at the end of this proof sketch.
In the remainder of the proof, we bound the weight of the balanced and imbalanced vertices separately.

\paragraph{Bounding the weight of balanced vertices.}
To bound the weight of balanced vertices, we use an entropy argument.
We consider the random variable corresponding to a uniformly random hypothesis from $[n]$.
On one hand, this random variable has entropy $\log n$.
On the other hand, we can take a uniformly random hypothesis from $[n]$ by an appropriate random walk down the decision tree. 
Starting from the root, at each vertex, we step to a child with probability proportional to the number of hypotheses in that child's subtree.
The total entropy of this process is given by $\sum_{v\in\mathcal{T}_G}^{} p(v)\mathbb{H}(v)$, where $\mathbb{H}(v)\defeq \mathbb{H}(\frac{p(v^-)}{p(v)})$ is the entropy of the random walk's step at $v$.
A simple argument (Lemma~\ref{lem:type-2}) shows that, for all balanced vertices $v$, we have $p(v^-)/p(v) \ge \delta/2$ and hence $\mathbb{H}(v) \ge \mathbb{H}(\delta/2) \ge \frac{\delta}{2}\log\frac{2}{\delta}$.
We thus have
\begin{align}
  \log n\ &= \   \sum_{v\in\mathcal{T}_G}^{} p(v)\mathbb{H}(v) \ge \sum_{v\text{ balanced}}^{} p(v)\mathbb{H}(v) \ge \frac{\delta}{2}\log\left(\frac{2}{\delta}\right)\sum_{v\text{ balanced}}^{} p(v).
\end{align}
Hence,
\begin{align}
  \sum_{v\text{ balanced}}^{} p(v) \ &\le \ \frac{\log n}{\log\frac{2}{\delta}}\cdot\frac{2}{\delta}.
\end{align}

\paragraph{Bounding the weight of imbalanced vertices.}
To bound the cost of imbalanced vertices, we crucially use a connection to \textsc{Min~Sum~Set~Cover} (MSSC).
In MSSC, one is given a universe $S$ and sets $A_1,\dots,A_M$, and needs to construct an ordering $\sigma:[M]\to[M]$ of the sets that minimizes the \emph{cost}: the cost of a solution $\sigma$ is the average of the \emph{cover times} of the elements in the universe $S$.
That is, the cost $\MSSC(\sigma)$ of a solution $\sigma$ is
\begin{align}
   \MSSC(\sigma)
  \ \defeq \ \frac{1}{n}\sum_{h\in S}^{}  \min\{\ell: h\in A_{\sigma(\ell)}\}.
\label{}
\end{align}
A result by Feige, Lovasz, and Tetali shows that the greedy algorithm gives a 4 approximation of MSSC, and they show this is tight by proving that finding a $(4-\varepsilon)$ approximation of MSSC is NP-hard.
On the lower bound side, a connection between MSSC and \textsc{Decision~Tree} was already known: Chakaravarthy et al. \cite{chakaravarthy2007decision} proved that it is NP-hard to approximate \textsc{Uniform~Decision~Tree} with ratio between than $4-\varepsilon$ by a reduction to MSSC.
The key technical contribution of our work is showing that there is also a connection on the upper bound side.
Bounding the weight of imbalanced vertices works as follows.
\begin{enumerate}
  \item For each chain $P$, define a corresponding instance $\MSSC\ind{P}$
    (Definition~\ref{def:mssc}) of \\ \textsc{Min~Sum~Set~Cover} \emph{induced} by the chain $P=(P_1,\dots,P_{|P|})$ as follows:

\begin{itemize}
\item Universe $S\defeq L(P_1)$, the set of all hypotheses that are consistent with $P_1$.
\item For $j=1,\dots,m$, the set $A_j$ is the set of hypotheses in $S$ that give the minority answer of test $j$ with respect to hypotheses $S$. (See Figure~\ref{fig:mssc-set}).
\item For each $h=1,\dots,n$, a set $A_{m+h} \defeq \{h\}\cap S$. These tests are included for technical reasons.
\end{itemize}
Note we have a total of $m + n$ sets, so that a solution is a permutation $\sigma:[m+n]\to[m+n]$.
The sets $A_j$ for $j=1,\dots,m$ are chosen so that the second step below holds.

\item Prove that the weight of a chain $P$ is bounded by the cost of a greedy solution to MSSC$\ind{P}$ (Lemma~\ref{lem:mssc}), and hence, using a result of Feige, Lovasz, and Tetali (Theorem~\ref{thm:flt}), by 4 times the optimal cost of MSSC$\ind{P}$ (Corollary~\ref{lem:mssc-2}).
That is, there exists a greedy solution $\sigma_G$ to $\MSSC\ind{P}$ such that
\begin{align}
  \sum_{v\in P}^{} p(v) 
  \ \le \ \MSSC\ind{P}(\sigma_G)
  \ \le \ 4\MSSC\ind{P}(\sigma_\text{OPT})
\end{align}
This step is somewhat technical, as one must show that the greediness of the greedy decision tree $\mathcal{T}_G$ produces a greedy solution to $\MSSC\ind{P}$.
The choice of $\sigma_G$ is natural: for $\ell=1,\dots,|P|$, let $\sigma_G(\ell)$ be the index of the test used at vertex $P_\ell$ in the chain $P$ (see Figure~\ref{fig:mssc}).
However, showing that this $\sigma_G$ is in fact a greedy solution to $\MSSC\ind{P}$ is a subtle argument that depends on the carefully chosen definition of a chain.
\begin{figure}
\begin{minipage}[t]{.45\textwidth}
  \begin{center}
  \includegraphics[height=150px]{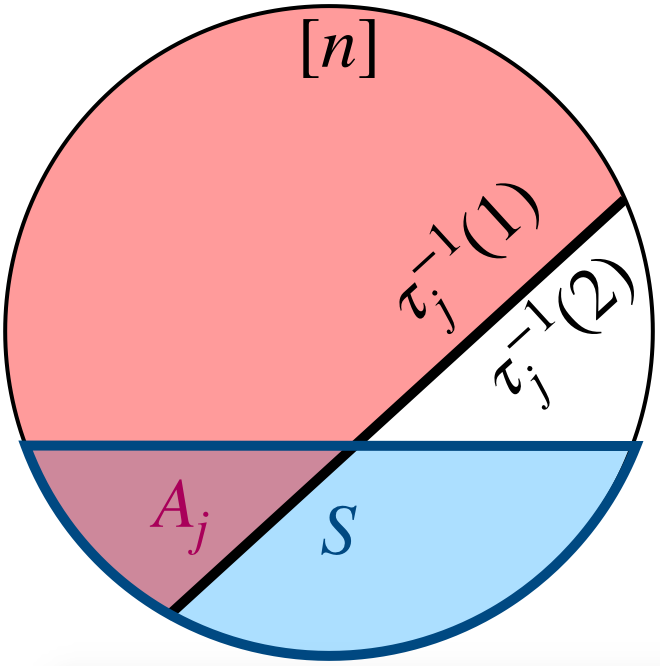}
  \end{center}
  \captionof{figure}{Illustration of set $A_j\subset S$, the part of test $j$'s partition of $S$ containing the smaller fraction of $S$.}
  \label{fig:mssc-set}
\end{minipage}
\begin{minipage}{0.04\textwidth}
\qquad
\end{minipage}
\begin{minipage}[t]{.5\textwidth}
  \includegraphics[height=150px]{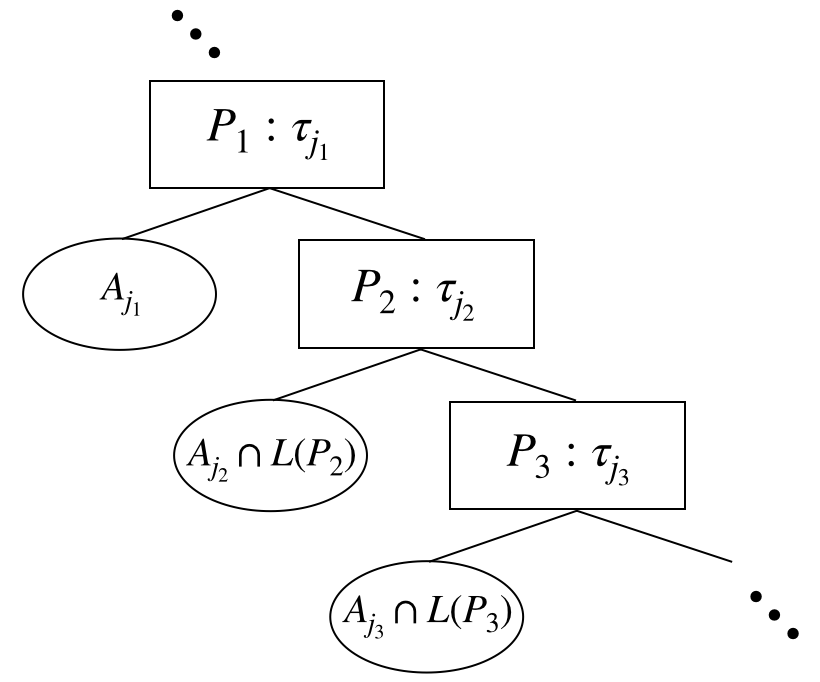}
  \captionof{figure}{MSSC sets $A_j$ in the greedy tree.
  Importantly (and subtly), sets $A_{j_1}, A_{j_2},\dots,$ cover the minority branches of the tree.
  We choose $\sigma_G$ by reading off the tests of the greedy tree: $\sigma_G(1)=j_1,\dots$.
  }
  \label{fig:mssc}
\end{minipage}
\end{figure}

\item Prove that, for any integer $s$, the sum, over all level-$s$ chains $P$, of optimal cost of MSSC$\ind{P}$, is bounded by $C_\text{OPT}$ (Lemma~\ref{lem:sum-2}).
Hence,
\begin{align}
  \sum_{P\in\mathcal{P}_s}^{} \MSSC\ind{P}(\sigma_\text{OPT})
  \ \le \ C_\text{OPT}.
\end{align}
This step is also technical, as one must draw the connection between the optimal MSSC solution and the optimal decision tree.

\item In total, we have
\begin{align}
  \sum_{v\text{ level $s$}}^{} p(v)
  \ &= \ \sum_{P\in\mathcal{P}_s}^{} \sum_{v\in P}^{} p(v)
  \ \le \ \sum_{P\in\mathcal{P}_s} 4\MSSC\ind{P}(\sigma_\text{OPT})
  \ \le \  4C_\text{OPT},
\end{align}
where the first inequality is by part 2 and the second inequality is by part 3.
In other words, for any integer $s$, the sum of the weights of all level $s$ chains is at most $4C_\text{OPT}$.
Hence, the sum of the weights of vertices in any chain, and thus the total weight of all imbalanced vertices, is at most $4s_\text{max}C_\text{OPT}$ (Lemma~\ref{lem:sum-3}), where $s_\text{max}$ is the number of levels.
As $s_\text{max}\le \frac{\log n}{\log\frac{1}{\delta}}$, we have
\begin{align}
  \sum_{v\text{ imbalanced}}^{} p(v) \ \le \ \frac{4\log n}{\log \frac{1}{\delta}}\cdot C_\text{OPT}.
\label{}
\end{align}
\end{enumerate}

To finish the proof, we bound
\begin{align}
  C_G
  \ = \  \sum_{v\in\mathcal{T}_G^\mathrm{o}}^{} p(v)
  \ = \  \sum_{v\text{ balanced}}^{} p(v) + \sum_{v\text{ imbalanced}}^{} p(v)
  \ \le \  \frac{\log n}{\log \frac{2}{\delta}}\cdot \frac{2}{\delta} + \frac{4\log n}{\log \frac{1}{\delta}}\cdot C_\text{OPT}.
\end{align}
The above is optimized roughly when $\delta=\frac{1}{C_\text{OPT}}$, giving the desired bound of $C_G\le \frac{6\log n}{\log C_\text{OPT}} \cdot C_\text{OPT}$.
If $n$ is sufficiently large, taking $\delta=\frac{10}{\varepsilon C_\text{OPT}}$ yields $C_G\le \frac{(4+\varepsilon)\log n}{\log C_{\OPT}}\cdot C_\text{OPT}$.

\subsection{General weights and larger $K$}
The proof of the general Theorem~\ref{thm:alg-main} follows similarly to the specific case given above.
The two differences are that Theorem~\ref{thm:alg-main} is stated for general $K$ and for general, not-necessarily-uniform distributions $p_1,\dots,p_n$.

Adapting the proof to general $K$ is the easier step.
The main difference is the definition of an \emph{imbalanced vertex}.
Now, we say a vertex $v$ is imbalanced if there is an integer $s$ such that $p(v)>2\delta^s$ and $p^-(v)\le \delta^s$, where $p^-(v)$ is the total weight of hypotheses in the subtrees of \emph{all children of $v$ except the majority vertex, $v^+$}, the child of $v$ with the largest weight of hypotheses.
Under this definition, a similar analysis follows.
Note that $p^-(v)$ could be much larger than $p(v^+)$ in this case, but this does not affect the proof much.
A little more care needed in the entropy argument for balanced vertices, and with the MSSC instance defined by a path $P$ now taking $A_j$ to be all hypotheses that do not take the majority answer of $\tau_j$ with respect to the MSSC universe.
Note that, if we specialize to $K=2$, the value $p^-(v)$ is simply $p(v^-)$.

In the weighted case, we again define $v$ to be imbalanced if there is an integer $s$ such that $p(v)>2\delta^s$ and $p^-(v)\le \delta^s$.
We again bound the cost of the balanced vertices by an entropy argument, and the cost of the imbalanced vertices via a connection to Min-Sum-Set-Cover.
However, because the entities are now weighted, we need to consider the greedy algorithm for a weighted generalization of MSSC called Weighted Min-Sum-Set-Cover (WMSSC).
In order to make the condition between the greedy decision tree and the greedy solution to WMSSC, we need a somewhat technical definition:
call a vertex $v$ is \emph{$h$-heavy} if $h$ is consistent with $v$ and $p_h>p^-(v)$.
Define $q(v) = p_h$ if there exists $h$ such that $v$ is $h$-heavy, and set $q(v)=0$ otherwise.
One can easily check that, for any vertex $v$, there is at most one $h$ such that $v$ is $h$-heavy, so $q(v)$ is well defined.
Now, we follow the argument in the uniform case, bounding
\begin{align}
  \sum_{v\text{ imbalanced}}^{} (p(v) - q(v))
  \ &\le \ \sum_{s=1}^{s_\text{max}} \sum_{P\in\mathcal{P}_s}^{} \sum_{v\in P}^{} (p(v) - q(v)) 
  \ \le \ \sum_{s=1}^{s_\text{max}} \sum_{P\in\mathcal{P}_s}^{} \WMSSC\ind{P}(\sigma_G\ind{P}) \nonumber \\
  \ &\le \ \sum_{s=1}^{s_\text{max}} \sum_{P\in\mathcal{P}_s}^{} 4\WMSSC\ind{P}(\sigma_\text{OPT}\ind{P}) 
  \ \le \ \sum_{s=1}^{s_\text{max}} C_\text{OPT}
  \ = \ s_\text{max}C_\text{OPT},
\end{align}
where $\sigma_G\ind{P}$ and $\sigma_\text{OPT}\ind{P}$ are the greedy solution and optimal solution, respectively, to the corresponding WMSSC.
The first inequality holds because $p(v) - q(v)\ge 0$ for all $v$ and every imbalanced vertex is in some chain\footnote{It is inequality because some imbalanced vertices may be in multiple chains}.
The second inequality holds by a technical lemma (Lemma~\ref{Z:lem:mssc}) comparing the greedy decision tree with a greedy solution to WMSSC.
Just as for MSSC, the greedy algorithm gives a 4 approximation for WMSSC, so the third inequality holds.
Additionally, for all $s=1,\dots,s_\text{max}$, we can still bound $\sum_{P\in\mathcal{P}_s}^{} \WMSSC\ind{P}(\sigma_\text{OPT})$, the sum of all WMSSC costs in a single level, by $C_\text{OPT}$, so the fourth inequality holds.
To finish:
\begin{align}
  \sum_{v\text{ imbalanced}}^{} p(v)
  \ \le \ s_\text{max}C_\text{OPT} + \sum_{v\in \mathcal{T}_G}^{} q(v)
  \ \le \ s_\text{max}C_\text{OPT} + \left( 1+\ln \frac{p_\text{max}}{p_\text{min}} \right) C_\text{OPT}.
\label{}
\end{align}
The last inequality (Lemma~\ref{Z:lem:alg-29}) comes from comparing, for fixed $h$, the vertices $v$ of the greedy tree $\mathcal{T}_G$ that are $h$-heavy to an appropriate SET-COVER instance, and using the fact that the greedy algorithm on a weighted generalization of SET-COVER gives a $1+\ln\left( \frac{p_\text{max}}{p_\text{min}}\right)$ approximation (Theorem~\ref{Z:thm:sc}).

\section{Sketch of proof of Theorem~\ref{thm:subexp-alg}}
\label{sec:subexp}
\subsection{Algorithm}
We describe the algorithm that achieves a $\frac{25}{\alpha} + \log R$ approximation for $\DT(R)$.
In Section~\ref{app:C}, we give the details and describe how the same algorithm with minor adjustments gives an improved approximation guarantee for \textsc{Uniform~Decision~Tree}.

The key idea in the algorithm is that, if the optimal tree has cost at least $n^{3\alpha/4}$, then the greedy algorithm gives an $O(1/\alpha)$ approximation by Theorem~\ref{thm:alg-main}. 
Fix $b\defeq \ceil{(12\log n +\log R)n^\alpha}$.
Our algorithm first computes the greedy tree.
If the cost of the greedy tree is at least $n^{-\alpha/4}b$, we simply return the greedy tree. 
Otherwise, we perform an exhaustive search over decision trees of depth at most $b$ such that all hypotheses not consistent with vertices at depth $b$ are uniquely distinguished.
We choose such a tree $\mathcal{T}$ with minimum cost (see definition of $C(\mathcal{T}; H)$ below).
Finally, at each leaf $v$ of $\mathcal{T}$ at depth $b$, we recursively compute a decision tree that distinguishes the hypotheses consistent with $v$.
The runtime of this algorithm is dominated by the exhaustive search, which we can solve in time $2^{\tilde O(n^{\alpha})}$ using a divide-and-conquer algorithm. 

Let $C(\mathcal{T}; H)$ denote the \emph{cost of a decision tree $\mathcal{T}$ with respect to hypothesis set $H$}, given by
\begin{align}
  C(\mathcal{T}; H) \defeq \frac{1}{p(H)}\sum_{h\in H}^{} p_h d_\mathcal{T}(h),
\label{}
\end{align}
where $d_\mathcal{T}(h)$ is the depth of the deepest vertex of $\mathcal{T}$ consistent with $h$.
In this way, we have $C(\mathcal{T}) = C(\mathcal{T}; [n])$.
To solve the \textsc{Decision~Tree} instance, we run Fulltree$_\alpha([n])$ below.
\begin{figure}[H]
  \begin{minipage}[t]{.48\linewidth}
    \begin{algorithm}[H]
      \caption{FullTree$_\alpha(H)$}
      \begin{algorithmic}
        \STATE Compute $ \mathcal{T}_G(H)$.
        \IF{ $p(H) < \frac{1}{n^2}$ }
          \STATE \textbf{Return} $\mathcal{T}_G(H)$
        \ENDIF
        \IF{ $C(\mathcal{T}_G(H);H) \geq n^{-\alpha/4}b$ }
          \STATE \textbf{Return} $\mathcal{T}_G(H)$
        \ENDIF
        \STATE $\mathcal{T}$ = PartialTree$(H,b)$.
        \FOR{$v \in \text{leaves}(\mathcal{T})$ with $d_\mathcal{T}(v) = b$}
          \STATE $H'=$ hypotheses in $H$ consistent with $v$.
          \STATE $\mathcal{T}_v = \text{FullTree}_\alpha(H')$
          \STATE Attach the root of $\mathcal{T}_v$ to $\mathcal{T}$ at vertex $v$.
        \ENDFOR
        \STATE \textbf{return} $\mathcal{T}$
      \end{algorithmic}
      \label{alg:fulltree}
    \end{algorithm}
  \end{minipage}
  \begin{minipage}[t]{.48\linewidth}
    \begin{algorithm}[H]
      \caption{PartialTree$(H,b')$}
      \begin{algorithmic}
        \IF{$b'=0$ or $|H| \leq 1$}
           \STATE \textbf{Return} leaf node
        \ENDIF
        \FOR{$j=1, \dots, m$}
          \FOR{$k=1, \dots, K$}
            \STATE $H_{j,k} = \{h \in H: \tau_j(h)=k\}$
            \STATE $\mathcal{T}_{j,k} = \text{PartialTree}(H_{j,k}, b'-1)$
          \ENDFOR
          \STATE Create tree $\mathcal{T}_j$ using $\tau_j$ at root and $\mathcal{T}_{j,k}$ as subtrees
        \ENDFOR
        \STATE $j^* = \argmin_j C(\mathcal{T}_j; H)$
        \STATE \textbf{Return} $\mathcal{T}_{j^*}$
        \STATE %
      \end{algorithmic}
      \label{alg:opt_trunc_tree}
    \end{algorithm}
  \end{minipage}
\end{figure}

\subsection{Analysis sketch}
We now sketch an analysis of the algorithm.
First, it is easy to check that FullTree$_\alpha([n])$ returns a valid decision tree.
By Theorem~\ref{thm:alg-main}, when the greedy tree is used in the recursive call FullTree$_\alpha(H)$, it gives an $O(1/\alpha) + \log R$ approximation to the \textsc{$\DT(R)$} instance induced by $H$.
Hence, by careful bookkeeping, the greedy trees included in the output tree contribute at most $(O(1/\alpha)+\log R)C_{OPT}$ to the cost (Lemma~\ref{lem:s-5}).
If the greedy tree is not used, then, in the optimal tree, the weighted average depth of the hypotheses is at most $n^{3\alpha/4}$. 
Hence, by a simple counting argument, at each recursive call, the fraction of undistinguished hypotheses shrinks by a factor of $n^{\alpha/4}$, so the maximum depth of recursive calls is $O(1/\alpha)$ (Lemma~\ref{lem:s-7}).
Careful bookkeeping shows that, for any $i=1,\dots,O(1/\alpha)$, the outputs to PartialTree$(H,b)$ called from the $i$th level of recursion collectively contribute at most $C_{\text{OPT}}$ to the cost of the output tree (Lemma~\ref{lem:s-6}).
Hence, the trees computed by exhaustive search across all levels of recursion contribute a cost of $O(C_{\text{OPT}}/\alpha)$.
Hence, the cost of our output tree is $(O(1/\alpha)+\log R)C_{OPT}$.

\section{Proof of Theorem~\ref{thm:alg-main} for uniform weights and $K=2$}

\label{sec:main-proof}

We prove a special case of Theorem~\ref{thm:alg-main} when $K=2$ and the weights $p_1=\cdots=p_n=\frac{1}{n}$ are uniform, that is, we show that the \textsc{Uniform~Decision~Tree} with binary tests gives an $O(\frac{\log n}{\log C_{\OPT}})$ approximation.
Throughout this section, we have a \textsc{Uniform~Decision~Tree} instance with hypotheses $[n]$ and tests $\tau_1,\dots,\tau_m:[n]\to[2]$.
\begin{theorem}
  \label{thm:alg-main-2}
  For any instance of the \textsc{Uniform~Decision~Tree} problem on $n$ hypotheses with branching factor  2, and any greedy tree $\mathcal{T}_G$ with average cost $C_G$, we have
  \begin{align}
    C_G
    \ &\le \ \frac{6\log n}{\log C_{\OPT}}\cdot C_{\OPT}. 
  \label{}
  \end{align}
\end{theorem}

\subsection{Notation}
\label{ssec:main-proof:not}

We use the following notation for our proof.
These notations help us reason about the greedy tree.
We write $v\in \mathcal{T}_G$ to mean that $v$ is a vertex of tree $\mathcal{T}_G$, and we write $v\in \mathcal{T}_G^\mathrm{o}$ to mean that $v$ is a interior vertex.
We say the \emph{length} of a path in the tree is the number of edges along the path.
For $v,u\in\mathcal{T}_G$, we say $v$ is an \emph{ancestor} of $u$ if there is a (possibly degenerate) path from $v$ to $u$ going down the tree.
In particular, $v$ is an ancestor of $v$.
We write this as $u\sqsubseteq v$.
We call $u$ a \emph{descendant} of $v$ if and only if $v$ is an ancestor of $u$.
For $v\in \mathcal{T}_G$, let $L(v)\subseteq[n]$ denote the set of hypotheses consistent with $v$.
For a subset $S\subseteq [n]$ of hypotheses, denote its \emph{weight} or \emph{cost} by $p(S) \defeq \sum_{h\in S}^{} p_h = \frac{|S|}{n}$. 
For brevity, let $p(v)\defeq p(L(v))$, denoting the \emph{weight} of vertex $v$, and we say the \emph{weight} of a set of vertices is the sum of the weights of the individual vertices in the set.

\subsection{The basic argument}
The following lemma shows that, rather than accounting the cost of the greedy tree by summing the depths of the leaves associated with the hypotheses, we can instead account the cost by summing the weights of vertices of the tree.
\begin{lemma}
  \label{lem:alg-10}
  We have $C_G =  \sum_{v\in \mathcal{T}_G^\mathrm{o}}^{} p(v)$. 
\end{lemma}
\begin{proof}
  We have,
  \begin{align}
    C_G
    \ &= \ \sum_{h\in[n]}^{} p_hd_{\mathcal{T}_G}(h) 
    \ = \ \sum_{h\in[n]}^{} \sum_{\substack{v\in\mathcal{T}_G^\mathrm{o} \\ h\sqsubseteq v}}^{} p_h
    \ = \ \sum_{v\in\mathcal{T}_G^\mathrm{o}}^{} \sum_{h\in L(v)}^{} p_h
    \ = \ \sum_{v\in\mathcal{T}_G^\mathrm{o}}^{} p(v),
  \end{align}
  where, in the third equality, we switched the order of summation.
\end{proof}

At a high level, our proof defines \emph{balanced} and \emph{imbalanced} vertices (next subsection) using a parameter $\delta$ and bound the weight of the balanced and imbalanced vertices separately.
We bound the weight of the balanced vertices by an entropy argument, and the weight of the imbalanced vertices by partitioning the imbalanced vertices into paths, called \emph{chains}, and bounding the weights of each chain separately.
Overall, we get the following bound.
  \begin{align}
    C_G
    \ = \  \sum_{v\in\mathcal{T}_G^\mathrm{o}}^{} p(v) 
    \ = \  \sum_{v\text{ balanced}}^{} p(v) + \sum_{v\text{ imbalanced}}^{} p(v)
    \ \le \  \frac{\log n}{\log \frac{2}{\delta}}\cdot \frac{2}{\delta} + \frac{4\log n}{\log \frac{1}{\delta}}\cdot C_{\OPT}.
  \end{align} 
  Choosing $\delta=\frac{1}{C_{\OPT}}$ gives $C_G\le \frac{6\log n}{\log C_{\OPT}} \cdot C_{\OPT}$.

For the rest of the proof, fix $\delta=\frac{1}{C_{\OPT}}$.
Additionally, for convenience and without loss of generality, assume that our instance is \emph{nontrivial}, i.e.\ there is some test $j$ such that both of $\tau^{-1}_j(1)$ and $\tau^{-1}_j(2)$ have at least 2 hypotheses, as otherwise the greedy tree is optimal and $C_G = C_{\OPT}$ and the theorem is true.

\subsection{More notation: Majority and minority answers}
We define majority (minority) answers, edges, children.
These definitions are useful for defining \emph{balanced} and \emph{imbalanced} vertices.
We later show that imbalanced vertices form paths whose edges are majority edges.
We call these paths \emph{chains}.
We then analyze the balanced and imbalanced vertices separately, and in particular analyze each path of majority edges separately.

\begin{figure}
\begin{center}
    \includegraphics[width=200px]{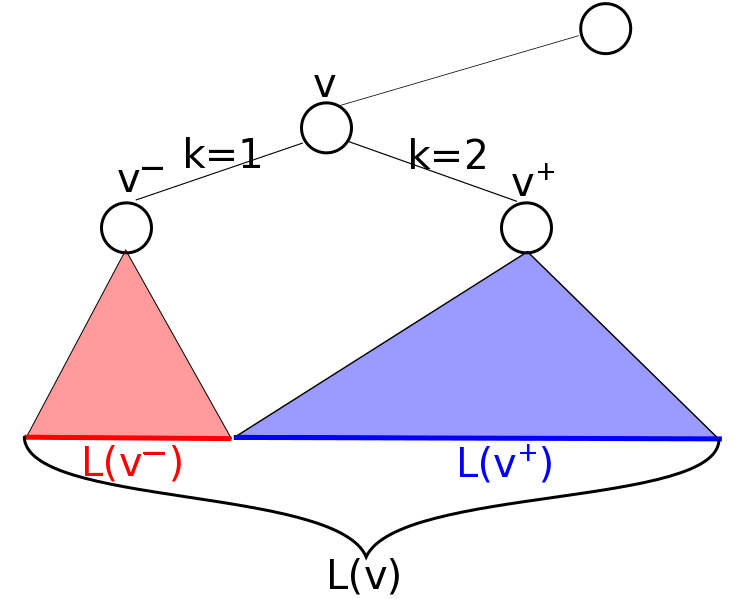}
    \caption{An illustration of the definitions of $v^-, v^+$ and $L(v)$.}
\label{fig:tree}
\end{center}
\end{figure}
For each vertex $v$ in the greedy tree, let $j_v$ denote the test used at $v$.
For each vertex $v$, label its children by $v^+$ and $v^-$ so that $p(v^+)\ge p(v^-)$, with ties broken\footnote{any tiebreaking procedure suffices, as long as the tiebreaking is consistent with the $k_{j,S}^+$ and $k_{j,S}^-$ notation in the next paragraph.} by labeling $v^+$ by the vertex corresponding to a test outputting 1.\footnote{it is possible to have a vertex that has one child, namely a test that doesn't distinguish any pairs of hypotheses at a vertex, but such a test is useless and never appears in either the greedy or optimal tree, so we assume such vertices don't exist.}
Accordingly, we have $p(v) = p(v^+) + p(v^-)$ for all $v\in\mathcal{T}_G^\mathrm{o}$.
Call the edge from $v$ to $v^+$ a \emph{majority edge}, and the edge from $v$ to $v^-$ a \emph{minority edge}. 
This is illustrated in Figure~\ref{fig:tree}.

In order to reason about the greedy tree precisely, we use the following notation which is more technical.
For test $j\in[m]$ and hypotheses $S\subseteq[n]$, let $k^+_{j,S}\in[2]$ be the answer to test $\tau_j$ that accounts for the maximum weight of hypotheses in $S$, and let $k^-_{j,S}$ be the other index, with ties broken by $k^+_{j,S}=1$.
In other words, $k^+_{j,S}$ and $k^-_{j,S}$ are chosen so that $|\tau_j^{-1}(k^+_{j,S})\cap S|\ge |\tau_j^{-1}(k^-_{j,S})\cap S|$.
We call $k^+_{j,S}$ the \emph{majority answer of test $j$ with respect to hypothesis set $S$}.
Call the other answer $k^-_{j,S}$ the \emph{minority answer of test $j$ with respect to hypothesis set $S$}.
For all $j\in[m]$ and $S\subseteq [n]$, let
\begin{align}
  I_{j,S}^+ \ \defeq \  \tau_j^{-1}(k^+_{j,S}),
  \qquad
  I_{j,S}^- \ \defeq \ \tau_j^{-1}(k^-_{j,S}).
\end{align}
We think of $I_{j,S}^+$ ($I_{j,S}^-$) as the set of hypotheses that, under test $j$, output the majority (minority) answer to test $j$ with respect to set $S$.
Note that, with the above notation, we have $L(v^+) = I_{j_v,L(v)}^+\cap L(v)$ and $L(v^-) = I_{j_v,L(v)}^-\cap L(v)$.

The following is a key property of the greedy tree $\mathcal{T}_G$: the weight of hypotheses consistent with the minority child $v^-$ decreases as we descend the tree.
\begin{lemma}
  For any vertices $u$ of $v$ with $u$ a descendant of $v$, we have $p(v^-)\ge p(u^-)$.
\label{lem:alg-9}
\end{lemma}
\begin{proof}
  Because $\mathcal{T}_G$ was constructed greedily, for all $v\in\mathcal{T}_G^\mathrm{o}$, the test $j_v$ was chosen to maximize the weight of $I_{j_v,L(v)}^-\cap L(v)$, the hypotheses in $L(v)$ giving the minority answer $k_{j_v,L(v)}^-$.
  Hence, any other test, in particular, the test $j_u$ chosen at vertex $u$, has a smaller weight of hypotheses of $L(v)$ that give the minority answer of $j_u$ with respect to hypotheses $L(v)$.
  Hence, we have $p(v^-) \ge p(I_{j_u,L(v)}^-\cap L(v))$.
  Hence,
  \begin{align}
    p(v^-)
    \ \ge \ p\left(I_{j_u,L(v)}^-\cap L(v)\right) 
    \ \ge \ p\left(I_{j_u,L(v)}^-\cap L(u)\right)  
    \ \ge \   p(u^-).
  \end{align}
  The second inequality holds because $L(u)\subseteq L(v)$.
  The third inequality holds because test $j_u$ defines a partition of $[n]$ into two parts, and $I_{j_u,L(v)}^-$ is one of the two parts, so $I_{j_u,L(v)}^-\cap L(u)$ is one of $L(u^-)$ or $L(u^+)$.
\end{proof}

\subsection{Defining balanced and imbalanced vertices}
In the following definition, we identify \emph{balanced} vertices and \emph{imbalanced} vertices.
By Lemma~\ref{lem:alg-10}, we can separately bound the weights of the balanced and imbalanced vertices.
\begin{definition}
  \label{def:type}
  Let $s$ be a positive integer.
  \begin{enumerate}
  \item We say a vertex $v\in\mathcal{T}_G^\mathrm{o}$ is \emph{level-$s$ imbalanced} if $p(v^-) \le \delta^s$ and $p(v) > 2\delta^s$.
  \item We say a vertex $v$ is \emph{imbalanced} if it is level-$s$ imbalanced for some $s$, and \emph{balanced} otherwise.
  \item We say a level-$s$ imbalanced vertex $v$ is \emph{minimal} if no descendant of $v$ is also level-$s$ imbalanced vertex, and a level-$s$ imbalanced vertex $v$ is \emph{maximal} if no ancestor of $v$ is level-$s$ imbalanced.
  \end{enumerate}
\end{definition}
Let 
\begin{align}
    s_{\max} \defeq \frac{\log n}{\log\frac{1}{\delta}}
\end{align}
and note that level-$s$ imbalanced vertices exist only for $s\le s_{\max}$.
The following lemma proves a structural result about balanced vertices, with the punchline being item (iii), which permits Definition~\ref{def:chain}.
For an illustration, see Figure~\ref{fig:chain}.
\begin{lemma}
  \label{lem:type-1}
  Let $s$ be a positive integer.
\begin{enumerate}
  \item[(i)] If $v$ is a level-$s$ imbalanced vertex, then, among the children of $v$, only $v^+$ can be a level-$s$ imbalanced vertex.
  \item[(ii)] Additionally, if $v$ and $u$ are level-$s$ imbalanced vertices and $v$ is an ancestor of $u$, then every vertex on the path from $v$ to $u$ is a level-$s$ imbalanced vertex.
  \item[(iii)] Finally, the set of level-$s$ imbalanced vertices can be partitioned into vertex disjoint paths, each of which connects a maximal level-$s$ imbalanced vertex to a minimal level-$s$ imbalanced vertex and contains only majority edges.
\end{enumerate}
\end{lemma}
\begin{proof}
  For (i), note that if $v$ is level-$s$ imbalanced, then $p(v^-)\le \delta^s < 2\delta^s$, so $v^-$ cannot be level-$s$ imbalanced.
  Hence, among the children of $v$, only $v^+$ can be level-$s$ imbalanced.

  For (ii), let $v\sqsupseteq w\sqsupseteq u$ be three vertices in the tree.
  Suppose that $v$ and $u$ are level-$s$ imbalanced.
  We know that $p(v) \ge p(w) \ge p(u) > 2\delta^s$, and Lemma~\ref{lem:alg-9} gives $\delta^s \ge p(v^-)\ge p(w^-)\ge p(u^-)$.
  Hence $w$ is level-$s$ imbalanced.
  
  For (iii), note that each level-$s$ imbalanced vertex has a maximal level-$s$ imbalanced ancestor (possibly itself), so we may partition the level-$s$ imbalanced vertices into sets based on their maximal level-$s$ imbalanced ancestor.
  We claim each set in the partition is a connected path.
  Let $v_1$ be the (unique) maximal level-$s$ imbalanced vertex in a set $P$.
  For $\ell=1,2,\dots$, if $v_\ell$ has a level-$s$ imbalanced child, let $v_{\ell+1}$ be that child, which is unique by the first item and in $P$ by definition.
  Let $\ell_P$ be the largest index such that $v_{\ell_P}$ is defined.
  Then $v_{\ell_P}$ has no level-$s$ imbalanced children.
  We claim $v_1,\dots,v_{\ell_P}$ are the only vertices in the set $P$.
  Suppose not.  Let $\ell$ be the largest index such that $v_{\ell}$ has a level-$s$ imbalanced descendant $u$ not among $v_1,\dots,v_{\ell_P}$.
  Then, by the second item, every vertex on the path from $v_{\ell}$ to $u$ is level-$s$ imbalanced.
  If $\ell=\ell_P$, this means $v_{\ell_P}$ has a level-$s$ imbalanced child, a contradiction of the maximality of $\ell_P$.
  Otherwise, as $\ell$ is maximal, $v_{\ell+1}$ is not on the path from $v_\ell$ to $u$, in which case, by (ii), $v_{\ell}^-$ is level-$s$ imbalanced, which contradicts (i).
  Thus, we always have a contradiction, so $P$ is the path $v_1,\dots,v_{\ell_P}$.
  By (i), every edge along $P$ is a majority edge.
  This completes the proof.
\end{proof}
\begin{figure}
\begin{center}
    \includegraphics[width=200px]{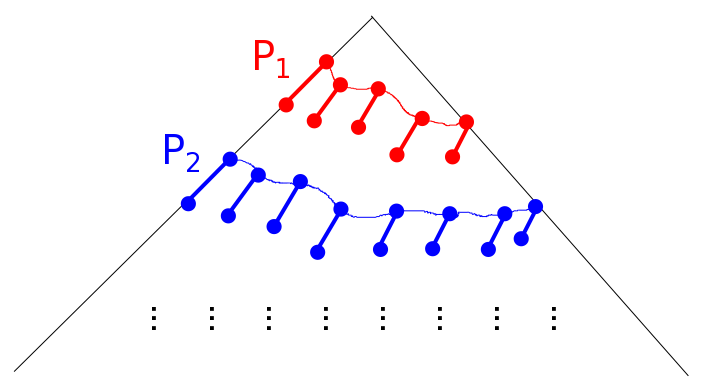} 
    \caption{
    Illustration of chains.
    The red vertices denote level-1 imbalanced vertices.
    The red segments denote level-1 chains.
    The blue vertices denote level-2 imbalanced vertices.
    The blue segments denote level-2 chains.
    The enlarged circular vertices denote maximal and minimal imbalanced vertices.
    }
    \label{fig:chain}
\end{center}
\end{figure}
Lemma~\ref{lem:type-1} motivates the following definition.
\begin{definition}
  \label{def:chain}
  Let $s$ be a positive integer.
  A \emph{level-$s$ chain}, $P=(P_1,\dots,P_{|P|})$, is a path of level-$s$ imbalanced vertices starting at a maximal level-$s$ imbalanced vertex and ending at a minimal level-$s$ imbalanced vertex.
  By Lemma~\ref{lem:type-1}, the level-$s$ chains partition the level-$s$ imbalanced vertices.
  We therefore let $\mathcal{P}_s$ denote the set of level-$s$ chains.
\end{definition}
In general, for $s\neq s'$, a level-$s$ chain might overlap with a level-$s'$ chain.

\subsection{Bounding the weight of balanced vertices}
We first prove a lemma that justifies the choice of the word ``balanced''.
\begin{lemma}
  For every balanced vertex $v$, we have $p(v^-)\ge \frac{\delta}{2}p(v)$.
  \label{lem:type-2}
\end{lemma}
\begin{proof}
Assume for contradiction that $p(v^-) < \frac{\delta}{2}p(v)$.
Let $s'>1$ be the real number such that $p(v^-) = \delta^{s'} n$.
In this way, $p(v^-) \le \delta^{\floor{s'}}n$.
Then
\begin{align}
  p(v) > \frac{2}{\delta}p(v^-) = \frac{2}{\delta}\delta^{s'}n = 2\delta^{s'-1}n > 2\delta^{\floor{s'}}n.
\end{align}
This implies that $v$ is level-$\floor{s'}$ imbalanced, so $v$ is imbalanced, a contradiction.
\end{proof}
We now bound the contribution of the balanced vertices to the weight using an entropy argument.
\begin{lemma}
\label{lem:type-3}
  We have
  \begin{align}
    \sum_{v\text{ balanced}}^{} p(v) \le \frac{\log n}{\log\frac{2}{\delta}}\cdot \frac{2}{\delta}.
  \label{}
  \end{align}
\end{lemma}
\begin{proof}
  For a vertex $v$ with a test of index $j$, let $X_v$ denote the binary random variable equal to $\tau_j(h)$ for an hypothesis $h$ chosen uniformly at random from the hypotheses of $L(v)$.
  Let $\mathbb{H}(\cdot)$ denote the entropy of a random variable.
  The entropy of a uniformly random hypothesis in $[n]$ equals $\log n$.
  On the other hand, we can pick a uniformly random hypothesis in $[n]$ by starting at the root vertex $v$, sampling an answer $X_v\in[2]$ for the test at $v$, stepping to $u$, the child of $v$ corresponding to the chosen answer, and repeating with $v=u$, until we reach a leaf.
  In this process, at any vertex $v$, the probability of stepping to a child $u$ is exactly $\frac{p(u)}{p(v)}$.
  Hence, by a simple induction, the probability of reaching any vertex $v$ in the tree during this process is exactly $p(v)$. 
  The total entropy of this process is thus $\sum_{v\in\mathcal{T}_G}^{} p(v)\cdot \mathbb{H}(X_v)$, as $p(v)$ is the probability of reaching vertex $v$.
  For a balanced vertex $v$, Lemma~\ref{lem:type-2} implies $\frac{\delta}{2} \le p(v^-)\le \frac{1}{2}$.
  Hence,
  \begin{align}
    \mathbb{H}(X_v)
    \ = \ p(v^+)\log \frac{1}{p(v^+)} + p(v^-)\log \frac{1}{p(v^-)}
    \ \ge \ p(v^-)\log \frac{1}{p(v^-)}
    \ \ge \ \frac{\delta}{2}\log\frac{2}{\delta}. 
  \end{align}
  We conclude
  \begin{align}
    \log n
    \ &= \  \sum_{v\in\mathcal{T}_G}^{} p(v)\cdot \mathbb{H}(X_v) 
    \ \ge \ \sum_{v\text{ balanced}}^{} p(v) \cdot \mathbb{H}(X_v)
    \ \ge \ \sum_{v\text{ balanced}}^{} p(v)\cdot \frac{\delta}{2}\log \frac{2}{\delta},
  \label{}
  \end{align}
  and rearranging gives the desired result.
\end{proof}

\subsection{Bounding the weight of imbalanced vertices}

We now bound the weight of imbalanced vertices using a connection to MSSC.

\subsubsection{Defining Min Sum Set Cover}
Recall that $I_{j,S}^+ = \tau_j^{-1}(k_{j,S}^+)$ and $I_{j,S}^- = \tau_j^{-1}(k_{j,S}^-)$ for all $j\in[m]$ and $S\subseteq[n]$.
\begin{definition} 
  \label{def:mssc}
Let $\MSSC\ind{P}$ denote the instance of \textsc{Min~Sum~Set~Cover} that is \emph{induced} by the chain $P=(P_1,\dots,P_{|P|})$. 
This instance is given by
\begin{itemize}
\item universe $S\defeq L(P_1)$,
\item for $j=1,\dots,m$, sets $A_j \defeq I_{j,S}^-\cap S$, and
\item for each $h=1,\dots,n$, a singleton set $A_{m+h} \defeq \{h\}\cap S$.\footnote{Some of these sets are empty, but we include them for notational convenience.}
\end{itemize}
Note we have a total of $m + n$ sets.
A \emph{solution} to the instance $\MSSC\ind{P}$ is a permutation $\sigma:[m+n]\to [m+n]$ corresponding to an ordering of the sets, and the \emph{cost} $\MSSC\ind{P}(\sigma)$ of a solution is the average of the \emph{cover times} of the elements in the universe $S$.
Formally,
\begin{align}
  \label{eq:wmssc}
  \MSSC\ind{P}(\sigma)
  \ \defeq \ \frac{1}{n}\sum_{h\in S}^{}  \min\{\ell: h\in A_{\sigma(\ell)}\}
  \ = \ \sum_{\ell=1}^{m+n} p\left( S \setminus (A_{\sigma(1)}\cup\cdots\cup A_{\sigma(\ell-1)}) \right). 
\end{align}
  Note that the cost of any solution is finite, as each hypothesis $h\in L(v)$ is in some set $A_j$.
We sometimes refer to a solution $\sigma$ by the sets $A_{\sigma(1)},\dots,A_{\sigma(m+n)}$.
\end{definition}
\begin{remark}
  \label{rem:mssc}
Since the initial \textsc{Uniform~Decision~Tree} instance always has a solution, any two hypotheses can be distinguished by one of the $m$ tests.
Hence, there is at most one hypothesis $h\in S$ such that, for all $j=1,\dots,m$, we have $h\notin A_j$. In other words, all but one of the sets $A_{m+h}$ for $h\in[n]$ is not used.
\end{remark}
\begin{definition}
\label{def:greedy}
We say a solution $\sigma:[m+n]\to[m+n]$ to $\MSSC\ind{P}$ is \emph{greedy at index $\ell$} if the set $A_{\sigma(\ell)}$ covers the maximum number of elements not covered by sets $A_{\sigma(1)},\dots, A_{\sigma(\ell-1)}$.
We say a solution $\sigma:[m+n]\to[m+n]$ is \emph{greedy} if it is greedy at index $\ell$ for all $\ell\in[m+n]$,
\end{definition}
Note that, in the case of ties, there may be multiple greedy solutions to $\MSSC\ind{P}$. 
Note also that, for any partial assignment $\sigma(1),\dots,\sigma(\ell)$, one can always complete the solution greedily, so that $\sigma:[m+n]\to[m+n]$ is greedy at indices $\ell+1,\ell+2,\dots,m+n$.
Definition~\ref{def:greedy} lets us leverage the following theorem, due to Feige, Lov\'asz, and Tetali. 
\begin{theorem}[Theorem 1 of \citep{FeigeLT04}]
  \label{thm:flt}
  The greedy algorithm gives a 4-approximation to the MSSC problem.
  Formally, let $\sigma_G$ be any greedy solution to the instance $\MSSC\ind{P}$, and let $\sigma_{\OPT}$ denote an optimal solution to $\MSSC\ind{P}$.
  We have
  \begin{align}
    \MSSC\ind{P}(\sigma_G) \le 4\cdot \MSSC\ind{P}(\sigma_{\OPT}).
  \label{}
  \end{align}
\end{theorem}

\subsubsection{Bounding chain weight above by MSSC cost}
The section shows that the weight of a chain $P$ is bounded by the cost of its corresponding MSSC instance.
To do this, we need the following technical lemma which shows that following the choices of the greedy tree yields a greedy solution to the MSSC, and hence the two weights are comparable.
\begin{lemma}
  Let $s$ be a positive integer and let $P$ be a level-$s$ chain.
  Then there exists a greedy solution $\sigma:[m+n]\to[m+n]$ to $\MSSC\ind{P}$, such that
  \label{lem:mssc}
  \begin{align}
    \sum_{v\in P}^{} p(v)
    \le
    \MSSC\ind{P}(\sigma) .
  \label{}
  \end{align}
\end{lemma}
\begin{proof}
  Let $P=(P_1,\dots,P_{|P|})$.
  Let $S=S\ind{P}$ be the universe and $A_1,\dots, A_{m+n}$ be the sets of the instance $\MSSC\ind{P}$.
  For $\ell=1,\dots,|P|$, let $j_\ell$ be the test used at vertex $P_\ell$.
  Define a solution $\sigma:[m+n]\to[m+n]$ to $\MSSC\ind{P}$ by setting $\sigma(\ell)=j_\ell$ for $\ell=1,\dots,|P|$, and completing the solution greedily.
  We claim $\sigma$ is a greedy solution.
  To prove this, we show the following. 
\begin{enumerate}
  \item[(i)] For all $j\in[m]$ and $1\le\ell\le |P|$, the majority answer for test $j$ with respect to $S=L(P_1)$ is the same as the majority answer for test $j$ with respect to $L(P_\ell)$.
  Equivalently, for all $j\in[m]$ and $1\le\ell\le |P|$, we have $A_j=I_{j,S}^-\cap S = I_{j,L(P_\ell)}^-\cap S$.
  As an immediate consequence, we know $A_{j_\ell}$ contains all of $L(P_\ell^-)$ and none of $L(P_\ell^+)$.
  \item[(ii)] The set of hypotheses of $S$ not covered by $A_{j_1}, \dots,A_{j_{\ell-1}}$ is exactly $L(P_\ell)$.
  \item[(iii)] For each $1\le \ell\le |P|$, among sets $A_1,\dots,A_{m+n}$, the set $A_{j_\ell}$ covers the maximum number of hypotheses in $L(P_\ell)$. i.e. we have
  \begin{align}
    p(A_{j_{\ell}}\cap L(P_\ell))=\max_{j\in[m+n]}p(A_j\cap L(P_\ell)).
  \end{align}
\end{enumerate}
  These points suffice, as (ii) and (iii) tell us that $\sigma$ is greedy at indices $1,\dots,|P|$, so by construction $\sigma$ is greedy.

  To show (i), fix $j\in[m]$ and $\ell\in \{1,\dots,|P|\}$.
  As $P_\ell$ is level-$s$ imbalanced, we also have $p(P_\ell)>2\delta^s$ and $p(P^-_\ell) \le \delta^s$ and $p(P^+_\ell) > \delta^s$, so $k_{j,L(P_\ell)}^+$ accounts for more than half of the hypotheses in $L(P_\ell)$.
  On the other hand, as $P_1$ is level-$s$ imbalanced, we have $p(I_{j,S}^-\cap L(P_\ell)) \le p(I_{j,S}^-\cap S) \le p(P^-_1) \le \delta^s$, so the majority answer for test $j$ with respect to hypothesis set $S$ also accounts for more than half of the hypotheses in $L(P_\ell)$.
  Hence $k_{j,S}^+ = k_{j,L(P_\ell)}^+$ for all $j\in[m]$ and $1\le\ell\le |P|$.

  Item (ii) follows because $L(P_\ell)$ is the set of hypotheses consistent with $P_\ell$, which was obtained by following the majority edges from $P_1$.
  This means $L(P_\ell)$ contains all the hypotheses of $S$ not consistent with a minority child of one of $P_1,\dots,P_{\ell-1}$.
  By (i), this is exactly $S\setminus (A_{j_1}\cup\cdots\cup A_{j_{\ell-1}})$.

  For (iii), at vertex $P_\ell$ in the greedy decision tree, the test index $j=j_\ell$ maximizes the weight $p(I_{j,L(P_\ell)}^-\cap L(P_\ell))$.
  By (i), this index $j$ equivalently maximizes $p(A_j\cap L(P_\ell))$, as desired.
  This completes the proof that $\sigma$ is greedy.

  We now return to the proof of Lemma~\ref{lem:mssc}.
  Take the greedy solution $\sigma$ given above.
  For $1\le\ell\le |P|$, the set of vertices of $S$ not covered by $A_{\sigma(1)},\dots,A_{\sigma(\ell-1)}$ is exactly $L(P_\ell)$, which has weight $p(P_\ell)$.
  Hence, by \eqref{eq:wmssc},
  \begin{align}
    \MSSC\ind{P}(\sigma)
    \ &\ge \ \sum_{\ell=1}^{|P|}  p\left( S\setminus(A_{\sigma(1)}\cup\cdots\cup A_{\sigma(\ell-1)}) \right)
    \ = \  \sum_{\ell=1}^{|P|} p(P_\ell),
  \label{}
  \end{align}
  as desired.
\end{proof}

By Theorem~\ref{thm:flt}, we have the following immediate corollary.
\begin{corollary}
\label{lem:mssc-2}
  Let $P$ be any chain, and $\sigma_{\OPT}$ be the optimal solution to $\MSSC\ind{P}$.
  Then
  \begin{align}
    \sum_{v \in P}^{} p(v)
    \le 4\MSSC\ind{P}(\sigma_{\OPT}) .
  \label{}
  \end{align}
\end{corollary}

\subsubsection{Bounding MSSC cost above by $C_{\OPT}$}
We now show that the optimal MSSC solution can be compared to the optimal decision tree cost, $C_{\OPT}$.
For all chains $P$, let $\sigma_G\ind{P}:[m+n]\to[m+n]$ be a greedy solution to $\MSSC\ind{P}$ given by Lemma~\ref{lem:mssc}, and let $\sigma_{\OPT}\ind{P}$ be an optimal solution to $\MSSC\ind{P}$.

\begin{lemma}
  \label{lem:sum-2}
  Let $s$ be a positive integer.
  We have
  \begin{align}
    \sum_{P\in \mathcal{P}_s}^{} \MSSC\ind{P}(\sigma_{\OPT}\ind{P})
    \le C_{\OPT}.
  \label{}
  \end{align}
\end{lemma}
\begin{proof}
  Let $S\ind{P}$ be the universe of the instance $\MSSC\ind{P}$, and let $A_1,\dots, A_{m+n}$ be the sets.
  Construct a path $w_1,\dots,w_{\ell\sar}$ in $\mathcal{T}_{\OPT}$ such that $w_1$ is the root and $w_{\ell\sar}$ is a leaf, which is identified with some hypothesis $h\sar$, and, for $\ell=1,\dots,\ell\sar-1$, if the test at vertex $w_\ell$ has index $j_\ell$, the edge to its child $w_{\ell+1}$ corresponds to the answer $k_{j_\ell,S\ind{P}}^+$, the majority answer of test $j_\ell$ with respect to set $S\ind{P}$.
  Since we follow the edges with label $k^+_{j_\ell,S\ind{P}}$, this corresponds to following the path for an hypothesis contained in $I_{j_\ell,S\ind{P}}^+$. 
  In other words, we have, for $\ell=1,\dots,\ell\sar-1$,
  \begin{align}
    L(w_{\ell+1}) \cap S\ind{P} 
    \ = \ (L(w_\ell)\cap S\ind{P}) \setminus I_{j_{\ell},S\ind{P}}^-
    \ = \ (L(w_\ell)\cap S\ind{P}) \setminus A_{j_{\ell}}. 
  \label{}
  \end{align}
  Thus the sequence $A_{j_1},\dots,A_{j_{\ell\sar}}, A_{m+h\sar}$ covers $S\ind{P}$, and thus gives a valid solution $\sigma_{\text{TREE}}$ to the instance $\MSSC\ind{P}$, where $\sigma_{\text{TREE}}(1) = j_1, \sigma_{\text{TREE}}(2)=j_2,\dots,\sigma_{\text{TREE}}(\ell\sar)=j_{\ell\sar}, \sigma_{\text{TREE}}(\ell\sar+1)=m+h\sar$, and $\sigma_{\text{TREE}}$ on larger indices is arbitrarily chosen.
  Note that the depth of the leaf for hypothesis $h$ is at least the number of vertices of $w_1,\dots,w_{\ell\sar}$ that are on the root-to-leaf path of $h$, and this number is $\min\{\ell:h\in A_{j_\ell}\}$, except for $h^*$, in which case it is 1 smaller.
  Furthermore, for some $h\in S\ind{P}$, the depth of leaf $h$ is at least $\min\{\ell:h\in A_{j_\ell}\}+1$, because otherwise all the branches leaving the path have one leaf, which can only happen if our \textsc{Uniform~Decision~Tree} instance is trivial, and it is not trivial by assumption.
  Thus,
  \begin{align}
    \label{eq:sum-2-2}
    \frac{1}{n}\sum_{h\in S\ind{P}}^{} d_{\mathcal{T}_{\OPT}}(h) 
    \ &\ge  \ \frac{-1}{n} + \frac{1}{n} + \frac{1}{n}\sum_{h\in S\ind{P}}^{} \min\{\ell:h\in A_{j_\ell}\} \nonumber\\
    \ &\ge \ \MSSC\ind{P}(\sigma_{\text{TREE}})  
    \ \ge \ \MSSC\ind{P}(\sigma_{\OPT}).
  \end{align}
  The $\frac{-1}{n}$ and $\frac{1}{n}$ account for the lower order terms described above.
  Summing \eqref{eq:sum-2-2} over $P\in\mathcal{P}_s$ gives
  \begin{align}
    C_{\OPT} \ &= \ \frac{1}{n}\sum_{h\in[n]}^{} d_{\mathcal{T}_{\OPT}}(h)
    \ \ge \ \frac{1}{n}\sum_{P\in\mathcal{P}_s}^{} \sum_{h\in S\ind{P}}^{} d_{\mathcal{T}_{\OPT}}(h)
    \ \ge \ \sum_{P\in\mathcal{P}_s}^{} \MSSC\ind{P}(\sigma_{\OPT}), 
  \label{}
  \end{align}
  where in the first inequality, we used that every leaf has at most one maximal level-$s$ imbalanced ancestor, and hence it is in at most one MSSC universe $S\ind{P}$.
\end{proof}

\subsubsection{Bounding imbalanced vertex weight above by $C_{\OPT}$}
We now finish our bound of the weight of imbalanced vertices.
\begin{lemma}
\label{lem:sum-3}
  We have 
  \begin{align}
    \sum_{v\text{ imbalanced}}^{} p(v)
    \ \le \ \frac{4\log n}{\log\frac{1}{\delta}}C_{\OPT}.
  \label{}
  \end{align}
\end{lemma}
\begin{proof}
  Each imbalanced vertex $v$ is level-$s$ imbalanced for some positive integer $s$, so it is part of some level-$s$ chain, $P$.
  Hence,
  \begin{align}
    \label{eq:sum-5}
    \sum_{v\text{ imbalanced}}^{} p(v) 
    \ \le \ \sum_{s=1}^{s_{\max}} \sum_{P\in\mathcal{P}_s}^{} \sum_{v \in P}^{} p(v) 
    \ &\le \ \sum_{s=1}^{s_{\max}} \sum_{P\in\mathcal{P}_s}^{}  4\MSSC\ind{P}(\sigma\ind{P}_{\OPT}) \nonumber\\
    \ &\le \ \sum_{s=1}^{s_{\max}} 4C_{\OPT} 
    \ = \  \frac{4\log n}{\log \frac{1}{\delta}}C_{\OPT} 
  \end{align}
  The first inequality is not equality because some vertices may be level-$s$ imbalanced for more than one integer $s$.
  The second inequality is by Corollary~\ref{lem:mssc-2}.
  The third inequality is by Lemma~\ref{lem:sum-2}.
\end{proof}

\subsection{Finishing the proof}
\begin{proof}[Proof of Theorem~\ref{thm:alg-main}]
  We have
  \begin{align}
    \label{eq:final}
    C_G
    \ = \  \sum_{v\in\mathcal{T}_G^\mathrm{o}}^{} p(v) 
    \ &= \  \sum_{v\text{ balanced}}^{} p(v) + \sum_{v\text{ imbalanced}}^{} p(v)
    \ \le \  \frac{\log n}{\log \frac{2}{\delta}}\cdot \frac{2}{\delta} + \frac{4\log n}{\log\frac{1}{\delta}}C_{\OPT} \nonumber\\
    \ &= \  \frac{2\log n}{\log 2C_{\OPT}}\cdot C_{\OPT} + \frac{4\log n}{\log C_{\OPT}} C_{\OPT}
    \ < \   \frac{6\log n}{\log C_{\OPT}} \cdot C_{\OPT},
  \end{align} 
  as desired. 
  In the first inequality, we used Lemma~\ref{lem:type-3} and Lemma~\ref{lem:sum-3}.
\end{proof}

\section{Proof of Theorem~\ref{thm:subexp-alg}}
\label{app:C}
\newcommand\greedy{\text{greedy}}

For the entirety of this section, fix $\alpha<1$, $R\ge 1$, and an instance of $\DT(R)$.
We first analyze Algorithm~\ref{alg:fulltree}, showing that running FullTree$_\alpha([n])$ gives a $(\frac{25}{\alpha}+\log R)$ approximation for $\DT(R)$, and then describe how the algorithm can be modified to give a $\frac{9.01}{\alpha}$ approximation for \textsc{Uniform~Decision~Tree} in subexponential time.

\subsection{Runtime}
\begin{lemma}
Algorithm \ref{alg:opt_trunc_tree} runs in time $O(n (Km)^{b+1})$.
\end{lemma}
\begin{proof}
  Each call to Algorithm~\ref{alg:opt_trunc_tree} calls at most $mK$ recursive calls with one less depth. This means the total number of recursive calls is $O((Km)^b)$. Since the local runtime of each call is $O(nmK)$, the total runtime is $O(n (Km)^{b+1})$.
\end{proof}

\begin{lemma}
Algorithm \ref{alg:fulltree} runs in time $O(n^3 (Km)^{b+1}) = 2^{O(n^{\alpha} \log (Rn) \log m)}$.
\end{lemma}
\begin{proof}
  The cost is dominated by the cost of Algorithm~\ref{alg:opt_trunc_tree}. The depth of recursive calls is at most $n$ and the width of the recursive call tree is at most $n$, thus the total runtime is at most $n^2$ times the runtime of Algorithm \ref{alg:opt_trunc_tree}. Thus the runtime is $O(n^3 (Km)^{b+1})$ 
\end{proof}

\subsection{Notation}
To formally analyze the approximation guarantees of Algorithm~\ref{alg:fulltree}, we need to generalize some earlier definitions. 
We say a decision tree $\mathcal{T}$ is \emph{complete with respect to hypothesis set $H$ up to depth $b$} if, for all hypotheses $h\in H$, either there exists a leaf $v$ of $\mathcal{T}$ with $L(v)\cap H = \{h\}$, or $d_\mathcal{T}(h) \ge b$.
Note that $\mathcal{T}$ is \emph{complete} with respect to hypothesis set $H$ if it is complete up to depth $b$ for all $b$.

Given the tests $\tau_1,\dots,\tau_m:[n]\to[2]$ and a subset $H$ of hypotheses, let $\Phi_H$ be the \emph{$\DT(R)$ instance induced by $H$}. 
It is given by hypotheses $H$ and the tests $\tau_1,\dots,\tau_m: H \to [2]$ restricted to domain $H$.
It is easy to check that, for this instance, we indeed have $\frac{p_{\max}'}{p_{\min}'}\le R$.
We let $\mathcal{T}_G(H)$ denote a greedy tree for the instance $\Phi_H$.
We let $\mathcal{T}_{\OPT}(H)$ denote an optimal tree for instance $\Phi_H$.
We also define the \emph{optimal tree for $\Phi_H$ up to depth $b$} by
\begin{align}
\mathcal{T}_{\OPT}(H, b) \defeq  \argmin_{\substack{\text{$\mathcal{T}$ complete w.r.t.\ $H$}\\ \text{up to depth $b$}}} C(\mathcal{T} ; H).
\end{align}
Importantly, $\mathcal{T}_{\OPT}(H,b)$ is computable by a straightforward recursive algorithm in time $O(n(mK)^b)$ by Algorithm~\ref{alg:opt_trunc_tree}.
For convenience, let
\begin{align}
  C_G(H) \ &\defeq \ C(\mathcal{T}_G(H); H)
  \nonumber\\
  C_{\OPT}(H) \ &\defeq \ C(\mathcal{T}_{\OPT}(H); H)
  \nonumber\\
  C_{\OPT}(H, b) \ &\defeq \ C(\mathcal{T}_{\OPT}(H, b); H).
\end{align}
In this way, we have
\begin{align}
  C_{\OPT}(H, b) \ \le \ C_{\OPT}(H) \ \le \  C_G(H).   
\end{align}

\subsection{Approximation guarantee}
It is easy to see that PartialTree$(H,b)$ computes $\mathcal{T}_{\OPT}(H,b)$ (or one such tree, if there are several).
Let $\mathcal{F}_i$ be the family of hypothesis sets $H$ such that FullTree$_\alpha(H)$ is called at the $i$th level of recursion and the greedy tree $\mathcal{T}_G(H)$ is \emph{not} returned. We consider FullTree$_\alpha([n])$ to be the 0th level of recursion so that $\mathcal{F}_0 = \{[n]\}$.
Let $\mathcal{F}_{\greedy}$ denote the family of hypothesis sets $H$ such that FullTree$_\alpha(H)$ is called and $\mathcal{T}_G(H)$ is returned in that call.
Let $\mathcal{F}_{\greedy} = \mathcal{F}_{\greedy}\ind{1}\cup\mathcal{F}_{\greedy}\ind{2}$ be a partition such that $H\in \mathcal{F}_{\greedy}\ind{2}$ if and only if $p(H)<\frac{1}{n^2}$
We know that, for any $H$ and $H'$, if FullTree$_\alpha(H')$ is called recursively from FullTree$_\alpha(H)$, then $H'\subseteq H$.
Thus, under these definitions, we know that, for all $i$, the hypothesis sets of $\mathcal{F}_i$ are pairwise disjoint, and the hypothesis sets of $\mathcal{F}_{\greedy}$ are pairwise disjoint.

By a double-counting argument, the cost of the output tree is the weighted sum of the partial trees and greedy trees computed in the recursion.
Formally, if $\mathcal{T}_{out}$ is the output tree,
\begin{align}
  \label{eq:subexp-1}
    C(\mathcal{T}_{\text{out}},[n]) = \sum_i \sum_{H \in \mathcal{F}_i} p(H) C_{\OPT}(H, b) + \sum_{H \in \mathcal{F}_{\greedy}} p(H) C_G(H).
\end{align}
Our proof bounds the depth of the recursion, as well as the summand components. 
\begin{lemma}
\label{lem:disjoint_hypotheses_opt}
\label{lem:s-4}
Let $\{H_i\}$ be a collection of disjoint subsets of $[n]$. Then $\sum_i p(H_i) C_{\OPT}(H_i) \leq C_{\OPT}$
\end{lemma}
\begin{proof}
   Recall that $\mathcal{T}_{\OPT} \defeq \mathcal{T}_{\OPT}([n])$ is the optimal tree of the $\DT(R)$ instance.
   By optimality of $\mathcal{T}_{\OPT}(H_i)$ for the instance instance induced by $H_i$, we have  $C_{\OPT}(H_i) = C(\mathcal{T}_{\OPT}(H_i); H_i)\leq C(\mathcal{T}_{\OPT}; H_i)$.
   Hence,
   \begin{align}
   \sum_{i=1}^\ell p(H_i) C_{\OPT}(H_i)
       \ &\leq \ \sum_{i=1}^\ell p(H_i) C(\mathcal{T}_{\OPT}; H_i)
       \ = \  \sum_i \sum_{h \in H_i} p_h \cdot d_{\mathcal{T}_{\OPT}}(h) \\
       \ &\leq \ \sum_{h \in [n]} p_h \cdot d_{\mathcal{T}_{\OPT}}(h) 
       \ = \ C_{\OPT}. \qedhere
   \end{align}
\end{proof}

\begin{lemma}
\label{lem:s-5}
We have
\begin{align}
\sum_{H \in \mathcal{F}_{\greedy}} p(H) C_G(H) \leq \left(\frac{17}{\alpha}+\log R\right)\cdot C_{\OPT}
\end{align}
\end{lemma}
\begin{proof}
  For all hypothesis sets $H\in \mathcal{F}_{\greedy}\ind{1}$, Algorithm~\ref{alg:fulltree} guarantees that $C_G(H) \geq (12 \log n + \log R) n^{3\alpha/4}$.
  By Theorem~\ref{thm:alg-main}, the greedy algorithm gives a $(12\log n + \log R)$ approximation on the $\DT(R)$ instance induced by $H$.
  Hence, $C_{\OPT}(H)\ge n^{3\alpha/4}$ for all $H\in \mathcal{F}{\greedy}\ind{1}$.
  By Theorem \ref{thm:alg-main} again, for all $H\in\mathcal{F}{\greedy}\ind{1}$, we have
  \begin{align}
    C_G(H)
    \ &\le \  \left(\frac{12\log(|H|)}{\log C_{\OPT}(H)} + \frac{c}{\alpha}\right) C_{\OPT}(H)
    \ \le \ \left(\frac{16}{\alpha}+\log R\right) \cdot C_{\OPT}(H).
  \end{align}
  Hence,
  \begin{align}
    \label{eq:subexp-2}
      \sum_{H \in \mathcal{F}{\greedy}\ind{1}} p(H) C_G(H) 
      \ &\leq \ \left(\frac{16}{\alpha}+\log R\right) \sum_{H \in \mathcal{F}{\greedy}\ind{1}} p(H) C_{\OPT}(H)
      \ \le \ \left(\frac{16}{\alpha}+\log R\right) \cdot C_{\OPT}  \\
    \label{eq:subexp-3}
      \sum_{H \in \mathcal{F}{\greedy}\ind{2}} p(H) C_G(H) 
      \ &< \   \sum_{H \in \mathcal{F}{\greedy}\ind{2}} \frac{1}{n^2}\cdot n  < 1 < \frac{1}{\alpha}\cdot C_{\OPT}.
  \end{align}
  The last inequality is by Lemma~\ref{lem:disjoint_hypotheses_opt} and the fact that the $H\in\mathcal{F}_{\greedy}$ are disjoint.
  Adding \eqref{eq:subexp-2} and \eqref{eq:subexp-3} gives the desired result
\end{proof}

\begin{lemma}
\label{lem:s-6}
For all $i=0,1,\dots$, we have
\begin{align}
  \sum_{H \in \mathcal{F}_i} p(H) C_{\OPT}(H, b)\leq C_{\OPT}
\end{align}
\end{lemma}
\begin{proof}
  For any $H \in \mathcal{F}_i$, we have $C_{\OPT}(H, b) \leq C_{\OPT}(H)$.
  Summing over all $H \in \mathcal{F}_i$, we have
  \begin{align}
      \sum_{H \in \mathcal{F}_i} p(H) C_{\OPT}(H, b)
      &\le \sum_{H \in \mathcal{F}_i} p(H) C_{\OPT}(H) 
      \le C_{\OPT}
  \end{align}
  where the last inequality follows from Lemma~\ref{lem:disjoint_hypotheses_opt} and that $H\in\mathcal{F}_i$ are disjoint.
\end{proof}

\begin{lemma}
\label{lem:s-7}
The maximum recursion depth in Algorithm~\ref{alg:fulltree} is at most $8/\alpha$
\end{lemma}
\begin{proof}
We show by induction that, for all $i$,
\begin{align}
  \sum_{H\in\mathcal{F}_i}^{} p(H)\le n^{-i\alpha/4}.
\label{eq:C-10}
\end{align}
This suffices, as then, for $i=\lfloor 8/\alpha \rfloor +1$, we have $\sum_{H \in \mathcal{F}_i} p(H) < \frac{1}{n^2}$.
By Algorithm~\ref{alg:fulltree}, we have $p(H)\ge \frac{1}{n^2}$ for all $H\in\mathcal{F}_i$ (otherwise we take the greedy tree). 
Thus, the maximum depth of the recursion in Algorithm \ref{alg:fulltree} is less than $\lfloor 8/\alpha \rfloor$.

Note that equality holds in \eqref{eq:C-10} for $i=0$, so the base case is true.
For the inductive step, fix $i\ge 1$, and let $H\in\mathcal{F}_i$.
Note that
\begin{align}
    n^{-\alpha/4}b
    &\geq C_G(H) \ge C_{\OPT}(H) \ge C_{\OPT}(H, b)
\end{align}
Consider a random variable $X$ equal to the depth in tree $\mathcal{T}_{\OPT}(H, b)$ of a random hypothesis in $H$ where $h\in H$ is chosen with probability proportional to $p_h$.
By above, $\E[X]\le n^{-\alpha/4}b$.
Hence, by Markov's inequality, $\Pr[X\ge b] \le \frac{\E[X]}{b} \le n^{-\alpha/4}$.
Thus, the total weight of hypotheses in $H$ that are in the next recursive call, i.e.\ in $H'$ for some $H'\in \mathcal{F}_{i+1}$, is at most $n^{-\alpha/4}p(H)$.
This holds for any $H$, so we conclude
\begin{align}
    \sum_{H \in \mathcal{F}_{i}} p(H) &\geq n^{\alpha/4} \sum_{H' \in \mathcal{F}_{i+1}} p(H').
\end{align}
This completes the induction, proving the lemma.
\end{proof}

\begin{lemma}
  \label{lem:s-8}
Let $\mathcal{T}_{\text{out}}$ be the tree returned by FullTree$_\alpha([n])$.
Then $C(\mathcal{T}_\text{out}, [n]) \leq (\frac{25}{\alpha}+\log R) C_{\OPT}$.
\end{lemma}
\begin{proof}
By \eqref{eq:subexp-1} and Lemmas~\ref{lem:s-5}, \ref{lem:s-6}, and \ref{lem:s-7}, we have
\begin{align}
    C(\mathcal{T}_{\text{out}},[n]) &= \sum_{i=1}^{\lfloor 8/\alpha \rfloor} \sum_{H \in \mathcal{F}_i} p(H) C_{\OPT}(H, b) + \sum_{H \in \mathcal{F}_\text{greedy}} p(H) C_G(H) \\
    &\leq \left(\sum_{i=1}^{\lfloor 8/\alpha \rfloor} C_{\OPT} \right) + \left(\frac{17}{\alpha} + \log R\right) C_{\OPT} 
    \le \left(\frac{25}{\alpha} +\log R \right) C_{\OPT}.\qedhere
\end{align}
\end{proof}

\subsection{Uniform Decision Tree}
We now describe how to modify Algorithm~\ref{alg:fulltree} to give a $\frac{9+\varepsilon}{\alpha}$ approximation for \textsc{Uniform~Decision~Tree} in subexponential time.
By the remark at the end of Theorem~\ref{thm:alg-main}, for all $\varepsilon>0$ there exists an $n_0 = n_0(\varepsilon)$ such that for all $n\ge n_0$, the greedy algorithm gives a $\frac{(4+\frac{2}{3}\varepsilon)\log n}{\log C_{\OPT}}$ approximation on \textsc{Uniform~Decision~Tree}.
Hence, the following \emph{modified greedy algorithm} runs in polynomial time and gives a $\frac{(4+\frac{2}{3}\varepsilon)\log n}{\log C_{\OPT}}$ approximation: for $n\ge n_0$, run the greedy algorithm, and for $n < n_0$, compute the optimal tree by brute force in constant time.
For \textsc{Uniform~Decision~Tree}, set $b=\ceil{(4+\frac{2}{3}\varepsilon)\log(n)n^\alpha}$, use the modified greedy algorithm instead of the greedy algorithm, and return the output of the modified greedy algorithm if $C(\mathcal{T}_G(H); H)\ge n^{-\alpha/3}b$ (rather than $n^{-\alpha/4}b$) and keep the rest of Algorithm~\ref{alg:fulltree} the same.
Lemma~\ref{lem:s-4} still holds.
For uniform weights, we have $F_{\greedy}\ind{2}=\emptyset$, so $\mathcal{F}_{\greedy}\ind{1}=\mathcal{F}_{\greedy}$.
Similar to Lemma~\ref{lem:s-5}, we are guaranteed that $C_{\OPT}(H)\ge n^{2\alpha/3}$ for all $H\in \mathcal{F}_{\greedy}$, and thus
\begin{align}
  \sum_{H\in\mathcal{F}_{\greedy}}^{} p(H)C_G(H) \le \sum_{H\in\mathcal{F}_{\greedy}}^{} \frac{(4+\varepsilon)\log(|H|)}{\log(C_{\OPT}(H))}\cdot C_{\OPT}(H) \le \frac{6+\varepsilon}{\alpha}\cdot C_{\OPT}.
\end{align}
Lemma~\ref{lem:s-6} still holds.
In Lemma~\ref{lem:s-7} the maximum depth of recursion is now $3/\alpha$ as the weight of hypotheses at each recursive call shrinks by a factor of $n^{\alpha/3}$ and the weight of hypotheses at each nonempty level is at most $\frac{1}{n}$.
Hence, the cost of the output tree has a contribution of at most $\frac{6+\varepsilon}{\alpha}C_{\OPT}$ from the greedy trees and at most $\frac{3}{\alpha}C_{\OPT}$ from the outputs of the PartialTree, for a total cost of at most $\frac{9+\varepsilon}{\alpha}C_{\OPT}$.

\section{Related Work}
\label{sec:rel}

There have been several other works analyzing \textsc{Decision~Tree} and they analyze it in a variety of cases to achieve the gold standard $O(\log n)$. While we examined the case with $K$-ary tests and non-uniform weights, we assumed that the tests had equal costs. Other works \citep{guillory2009average,gupta2010approximation} analyze the case where the test costs are non-uniform. \cite{guillory2009average} shows that the greedy algorithm yields $O(\log n)$ when either the costs are non-uniform or the weights are non-uniform (with the rounding trick) but not both. \cite{gupta2010approximation} introduces a new algorithm that achieves $O(\log n)$ with both non-uniform weights and costs.

In this work we studied the average depth of decision trees. We remark that, in the \emph{worst-case} decision tree problem, where the cost of a tree is defined to be the \emph{maximum} depth of a leaf in the tree, the approximability is known.
The greedy algorithm gives an $O(\log n)$ approximation \citep{arkin1998decision}, and obtaining a $o(\log n)$ approximation is NP-hard \citep{laber2004hardness}.

For the worst-case decision tree problem, there is a line of work that examines the \emph{absolute} query rate rather than the query rate relative to the optimal. In this line of work, the chief goal is to identify conditions where the greedy algorithm achieves the information-theoretically optimal rate $O(\log n)$. 
One such condition that ensures the $O(\log n)$ rate is ``sample-rich'' \citep{naghshvar2012noisy},
which states that every binary partition of the hypotheses has a test with matching pre-images.
\cite{nowak2009noisy,nowak2011geometry} introduced the more lenient
\emph{$k$-neighborly} condition, which requires that every two tests be connected by a sequence of tests
where neighboring tests disagree on at most $k$ hypotheses. An even more general condition is the split-neighborly condition \citep{mussmann2018generalized}, which is satisfied if every two tests are connected by a sequence of tests where neighboring tests must have every subset of the disagreeing hypotheses be evenly split by some other test.

\section{Conclusion}
\label{sec:conclusion}
There are several open questions left by our work.
\begin{enumerate}
\item 
Could one prove hardness of approximation results for \textsc{Uniform~Decision~Tree} for ratios larger than $4-\varepsilon$? It would be interesting to prove either NP-hardness results for larger constant factor approximations, or fine-grained complexity results for larger approximation ratios such as in \cite{ManurangsiR17}. 
\item On the flip side, could one find faster, perhaps polynomial time algorithms for approximating \textsc{Uniform~Decision~Tree} for ratios where we now have subexponential time algorithms?
\item On can also consider a generalization of \textsc{Decision~Tree} when the test costs are non-uniform. \citep{guillory2009average,gupta2010approximation} Could one obtain similar results in this setting?
\end{enumerate}

\section{Acknowledgements}
The authors thank Joshua Brakensiek for helpful discussions and feedback on an earlier draft of this paper.
The authors thank Mary Wootters for helpful feedback on an earlier draft of this paper.
The authors thank  anonymous reviews for helpful feedback on an earlier draft of this paper.

\bibliographystyle{alpha}
\bibliography{bibliography}

\appendix

\section{Proof of Theorem~\ref{thm:alg-main}}
\label{app:Z}

We now give a proof of Theorem~\ref{thm:alg-main}, highlighting the differences with the proof of the special case in Section~\ref{sec:main-proof}, and suppressing parts of the proof that are identical.

\subsection{Notation}
We reuse all of the notation in Section~\ref{ssec:main-proof:not}.
The only difference is that, in this section, $p(S)\defeq\sum_{h\in S}^{} p_h$ is not necessarily equal to $\frac{|S|}{n}$.
Just as in Section~\ref{sec:main-proof}, fix $\delta=\frac{1}{C_{\OPT}}$.

\subsection{The basic argument}
Lemma~\ref{lem:alg-10} is still true, and we restate it for completeness.
\begin{lemma}[Lemma~\ref{lem:alg-10}, restated]
  \label{Z:lem:alg-10}
  We have $C_G =  \sum_{v\in \mathcal{T}_G^\mathrm{o}}^{} p(v)$. 
\end{lemma}

At a high level, our proof defines \emph{balanced} and \emph{imbalanced} vertices (next subsection) using the parameter $\delta$ and bound the weight of the balanced and imbalanced vertices separately.
We bound the weight of the balanced vertices by an entropy argument, and the weight of the imbalanced vertices by partitioning the imbalanced vertices into paths, called \emph{chains}, and bounding the weights of each chain separately.
If a vertex $v$ has a \emph{heavy} hypothesis $h$ (defined in Section~\ref{sssec:Z:heavy}), we set $q(v)=p_h$, and otherwise we set $q(v)=0$.
To bound the cost of imbalanced vertices, we also need to bound the costs of heavy hypotheses $q(v)$.
Overall, we make get the following bounds.
  \begin{align}
    C_G
    \ &= \  \sum_{v\in\mathcal{T}_G^\mathrm{o}}^{} p(v) \nonumber\\
    \ &= \  \sum_{v\text{ balanced}}^{} p(v) + \sum_{v\text{ imbalanced}}^{} p(v) \nonumber\\
    \ &\le \  \frac{\log n}{\log 2C_{\OPT}}\cdot 2C_{\OPT} + \left(4s_{\max}(C_{\OPT}+1) + \sum_{v\in\mathcal{T}_G^\mathrm{o}}^{} q(v)\right) \nonumber\\
    \ &\le \  \frac{\log n}{\log 2C_{\OPT}}\cdot 2C_{\OPT} + 4\cdot \frac{\log\frac{1}{p_{\min}}}{\log \frac{1}{\delta}} \cdot (C_{\OPT}+1) + \left( 1 + \ln\frac{p_{\max}}{p_{\min}} \right)\cdot C_{\OPT} \nonumber\\
    \ &< \   \left(\frac{12\log\frac{1}{p_{\min}}}{\log C_{\OPT}} + \ln\frac{p_{\max}}{p_{\min}}\right)\cdot C_{\OPT}.
  \end{align}

\subsection{More notation: Majority and minority answers}
Again, we define majority (minority) answers, edges, children, which are useful for defining \emph{balanced} and \emph{imbalanced} vertices.

For each vertex $v$ in the greedy tree, let $j_v$ denote the test used at $v$.
For each vertex $v$, label its children by $v^+$ and $v^{1,-}, v^{2,-},\dots,$ so that $p(v^+)\ge p(v^{\ell,-})$ for all $\ell$, with ties broken\footnote{any tiebreaking procedure suffices, as long as the tiebreaking is consistent with the $k_{j,S}^+$ and $k_{j,S}^-$ notation in the next paragraph.} by labeling $v^+$ by the vertex corresponding to the largest answer.\footnote{it is possible to have a vertex that has one child, namely a test that doesn't distinguish any pairs of hypotheses at a vertex, but such a test is useless and never appears in either the greedy or optimal tree, so we assume it doesn't exist.}
Call the edge from $v$ to $v^+$ a \emph{majority edge}\footnote{Here, we may have $p(v^+) < \frac{1}{2}p(v)$, so the weight of hypotheses consistent with $v^+$ do not necessarily constitute a majority. However, this does difference does not affect the proof, and we keep the wording to stay consistent with Section~\ref{sec:main-proof}.}, and the edges from $v$ to $v^{\ell,-}$ \emph{minority edges}. 
Call $v^+$ the \emph{minority child} of $v$ and call $v^{1,-},v^{2,-},\dots,$ the minority children of $v$.
Let $L^-(v)\defeq \cup_{\ell}L(v^{\ell,-})$ be the hypotheses consistent with the minority children of $v$, and let their weight be $p^-(v) = p(L^-(v))$.
Accordingly, we have $p(v) = p(v^+) + p^-(v)$ for all $v\in\mathcal{T}_G^\mathrm{o}$.
This is illustrated in Figure~\ref{fig:tree}.

In order to reason about the greedy tree precisely, we use the following notation which is more technical.
For test $j\in[m]$ and hypotheses $S\subseteq[n]$, let $k^+_{j,S}\in[K]$ be the answer to test $\tau_j$ that accounts for the maximum weight of hypotheses in $S$, with ties broken by choosing the largest indexed answer $k^+_{j,S}$.
We call $k^+_{j,S}$ the \emph{majority answer of test $j$ with respect to hypothesis set $S$}.
Call the other answers the \emph{minority answers of test $j$ with respect to hypothesis set $S$}.
For all $j\in[m]$ and $S\subset [n]$, let
\begin{align}
  I_{j,S}^+ \ \defeq \  \tau_j^{-1}(k^+_{j,S}),
  \qquad
  I_{j,S}^- \ \defeq \ [n] \setminus I_{j,S}^+.
\end{align}
We think of $I_{j,S}^+$ ($I_{j,S}^-$) as the set of hypotheses that, under test $j$, output the majority (minority) answer to test $j$ with respect to set $S$.
Note that, with the above notation, we have $L(v^+) = I_{j_v,L(v)}^+\cap L(v)$ and $L^-(v) = I_{j_v,L(v)}^-\cap L(v)$.
Under these more general definitions, a generalization of Lemma~\ref{lem:alg-9} holds.
The proof is identical to that of Lemma~\ref{lem:alg-9}, so we omit it.
\begin{lemma}
  \label{Z:lem:alg-9}
  For any vertices $u$ of $v$ with $u$ a descendant of $v$, we have $p^-(v)\ge p^-(u)$.
\end{lemma}

\subsection{Defining balanced and imbalanced vertices}
In the following definition, we identify \emph{balanced} vertices and \emph{imbalanced} vertices.
By Lemma~\ref{Z:lem:alg-10}, we can separately bound the weights of the balanced and imbalanced vertices.
\begin{definition}
  \label{Z:def:level}
  Let $s$ be a positive integer.
  \begin{enumerate}
  \item We say a vertex $v\in\mathcal{T}_G^\mathrm{o}$ is \emph{level-$s$ imbalanced} if $p^-(v) \le \delta^s$ and $p(v) > 2\delta^s$.
  \item We say a vertex $v$ is \emph{imbalanced} if it is level-$s$ imbalanced for some $s$, and \emph{balanced} otherwise.
  \item We say a level-$s$ imbalanced vertex $v$ is \emph{minimal} if no descendant of $v$ is also level-$s$ imbalanced vertex, and a level-$s$ imbalanced vertex $v$ is \emph{maximal} if no ancestor of $v$ is level-$s$ imbalanced.
  \end{enumerate}
\end{definition}
Let 
\begin{align}
    s_{\max} \defeq \frac{\log \frac{1}{p_{\min}}}{\log\frac{1}{\delta}}
\end{align}
and note that interior level-$s$ imbalanced vertices exist only for $s\le s_{\max}$.
The following lemma proves a structural result about balanced vertices, with the punchline being item (iii), which permits Definition~\ref{Z:def:chain}.
The proof of Lemma~\ref{Z:lem:level-1} is nearly identical to that of Lemma~\ref{lem:type-1}.
We include a proof of item (i) because of a subtle difference to the proof of item (i) of Lemma~\ref{lem:type-1}.
However, the proofs of the other two parts are identical, so we omit them.
\begin{lemma}
  \label{Z:lem:level-1}
  Let $s$ be a positive integer.
  \begin{enumerate}
    \item[(i)] If $v$ is a level-$s$ imbalanced vertex, then, among the children of $v$, only $v^+$ can be a level-$s$ imbalanced vertex.
    \item[(ii)] Additionally, if $v$ and $u$ are level-$s$ imbalanced vertices and $v$ is an ancestor of $u$, then every vertex on the path from $v$ to $u$ is a level-$s$ imbalanced vertex.
    \item[(iii)] Finally, the set of level-$s$ imbalanced vertices can be partitioned into vertex disjoint paths, each of which connects a maximal level-$s$ imbalanced vertex to a minimal level-$s$ imbalanced vertex and contains only majority edges.
  \end{enumerate}
\end{lemma}
\begin{proof}[Proof of (i)]
  Note that if $v$ is level-$s$ imbalanced, then $p^-(v)\le \delta^s$, which means every $u$ different from $v^+$ satisfies $p(u)\le p^-(v)\le \delta^s < 2\delta^s$, so such $u$ cannot be level-$s$ imbalanced.
  Hence, among the children of $v$, only $v^+$ can be level-$s$ imbalanced.
\end{proof}
Lemma~\ref{Z:lem:level-1} motivates the following definition.
\begin{definition}
  \label{Z:def:chain}
  Let $s$ be a positive integer.
  A \emph{level-$s$ chain}, $P=(P_1,\dots,P_{|P|})$, is a sequence of level-$s$ imbalanced vertices starting at a maximal level-$s$ imbalanced vertex and ending at a minimal level-$s$ imbalanced vertex.
  By Lemma~\ref{Z:lem:level-1}, the level-$s$ chains partition the level-$s$ imbalanced vertices.
  We therefore let $\mathcal{P}_s$ denote the level-$s$ chains.
\end{definition}
In general, for $s\neq s'$, a level-$s$ chain might overlap with a level-$s'$ chain.

\subsection{Bounding the weight of balanced vertices}
Under these definitions, a generalization of Lemma~\ref{lem:type-2} is still true.
The proof is identical to that of Lemma~\ref{lem:type-2}, so we omit it.
\begin{lemma}
  \label{Z:lem:level-2}
  For every balanced vertex $v$, we have $p^-(v)\ge \frac{\delta}{2}p(v)$.
\end{lemma}
We now bound the contribution of the balanced vertices to the weight using an entropy argument.
Now, in the general case, the entropy argument requires a little more care when bounding the entropy of a single $K$-ary test.
\begin{lemma}
\label{Z:lem:level-3}
  We have
  \begin{align}
    \sum_{v\text{ balanced}}^{} p(v) \le \frac{\log n}{\log\frac{2}{\delta}}\cdot \frac{2}{\delta}.
  \end{align}
\end{lemma}
\begin{proof}
  For a vertex $v$ with a test of index $j$, let $X_v$ denote the random variable supported on $[K]$ that is equal to $\tau_j(h)$ for an hypothesis $h$ chosen randomly from the elements of $L(v)$, where the probability of choosing $h$ is proportional to $p_h$.
  Let $\mathbb{H}(\cdot)$ denote the entropy of a random variable, and by abuse of notation, let $\mathbb{H}(v)\defeq \mathbb{H}(X_v)$.
  By abuse of notation, for nonnegative $\alpha_1,\dots,\alpha_K$ summing to 1, let $\mathbb{H}(\alpha_1,\dots,\alpha_K) = \sum_{i=1}^{K} -\alpha_i\log\alpha_i$ where $0\log 0$ is taken to be 0.
  The entropy of a random element $[n]$ chosen according to the prior distribution $\bfp=(p_1,\dots,p_n)$ is at most $\log n$.
  On the other hand, we can pick a random hypothesis in $[n]$ according to the distribution $\bfp$ by setting $v$ to the root of $\mathcal{T}_G$, sampling an answer $X_v\in[K]$ for the test at $v$, setting $v$ to the child of $v$ corresponding to the chosen answer $X_v$, and repeating, until we reach a leaf.
  In this process, at any vertex $v$, the probability of stepping to a child $u$ is exactly $\frac{p(u)}{p(v)}$.
  Hence, by a simple induction, the probability of reaching any vertex $v$ in the tree during this process is exactly $p(v)$. 
  The total entropy of this process is thus $\sum_{v\in\mathcal{T}_G}^{} p(v)\cdot \mathbb{H}(v)$, as $p(v)$ is the probability of reaching vertex $v$ and $\mathbb{H}(v)$ is the entropy of the random variable $X_v$ chosen at vertex $v$.

  Fix a balanced vertex $v$.
  We claim that $\mathbb{H}(v)\ge \frac{\delta}{2}\log \frac{2}{\delta}$.
  Let $\mathcal{R}\subset\mathbb{R}^K$ denote the region given by the constraints $0\le \alpha_k\le \alpha_1\le 1$ for all $k=2,\dots,K$, and $\alpha_1+\cdots+\alpha_K=1$, and $\alpha_2+\cdots+\alpha_K\ge \delta/2$. 
  We claim that the minimum of $\mathbb{H}(\alpha_1,\dots,\alpha_K)$ for $(\alpha_1,\dots,\alpha_K)\in\mathcal{R}$ is $\mathbb{H}(1-\delta/2,\delta/2)$.
  To see this, note first that this region $\mathcal{R}$ is closed and bounded, so the function $\mathbb{H}(\alpha_1,\dots,\alpha_k)$ obtains a minimum.
  Furthermore, note that, for $(\alpha_1,\dots,\alpha_K)\in\mathcal{R}$, by concavity of $-x\log x$, for any $k=2,\dots,K$ and any $\varepsilon\le \alpha_K$, setting $(\alpha_1',\dots,\alpha_K')=(\alpha_1+\varepsilon,\alpha_2,\dots,\alpha_{k-1},\alpha_k-\varepsilon,\alpha_{k+1},\dots,\alpha_K)$ gives $\mathbb{H}(\alpha_1,\dots,\alpha_K) > \mathbb{H}(\alpha_1',\dots,\alpha_K')$. 
  Similarly pushing $\alpha_{k}$ and $\alpha_{k'}$ apart by the same positive $\varepsilon$ also decreases the value of $H$.
  Hence, the maximum cannot be obtained when two of $\alpha_2,\dots,\alpha_K$ are positive, nor can it be obtained when $\alpha_2+\cdots+\alpha_K > \delta/2$.
  It follows that the only local minima in the region occur when some $\alpha_k$ is $\frac{\delta}{2}$ and $\alpha_1=1-\frac{\delta}{2}$.

  For $k=1,\dots,K$, let $\alpha_k$ denote the probability that $X_v=k$. 
  By Lemma~\ref{Z:lem:level-2}, when $v$ is balanced, $(\alpha_1,\dots,\alpha_K)$ must, up to a permutation in coordinates, be in region $\mathcal{R}$.
  Hence, by the above we have
  \begin{align}
    \mathbb{H}(v)
    \ = \ \mathbb{H}(X_v) 
    \ = \ \mathbb{H}(\alpha_1,\dots,\alpha_k)
    \ \ge \ H\left( \frac{\delta}{2}, 1-\frac{\delta}{2} \right) > \frac{\delta}{2}\log\frac{2}{\delta}. 
  \end{align}

  Putting the above two paragraphs together, we conclude
  \begin{align}
    \log n
    \ &\ge \  \sum_{v\in\mathcal{T}_G}^{} p(v)\cdot \mathbb{H}(v) 
    \ \ge \ \sum_{v\text{ balanced}}^{} p(v) \cdot \mathbb{H}(v)
    \ \ge \ \sum_{v\text{ balanced}}^{} p(v)\cdot \frac{\delta}{2}\log \frac{2}{\delta},
  \end{align}
  and rearranging gives the desired result.
\end{proof}

\subsection{Bounding the weight of imbalanced vertices}

We now bound the weight of imbalanced vertices using a connection to Weighted Min Sum Set Cover.
For each hypothesis $h$, let $u_h^\bot$ denote the leaf in the greedy tree $\mathcal{T}_G$ for which hypothesis $h$ is consistent.
Since $\mathcal{T}_G$ is complete, this leaf exists and is unique.

\subsubsection{Technical definition: heavy vertices}
\label{sssec:Z:heavy}
We need the following technical definition to make the connection between the greedy decision tree and a greedy WMSSC solution.
\begin{definition}
  For a vertex $v$ and an hypothesis $h\in[n]$, we say $v$ is \emph{$h$-heavy} if $h$ is consistent with $v$ and $p_h > p^-(v)$.
\end{definition}
\begin{lemma}
\label{Z:lem:heavy-1}
  Let $h\in[n]$ be a hypothesis.
  \begin{enumerate}
  \item [(i)] If $v$ is $h$-heavy, then every vertex on the path from $v$ to leaf $u_h^\bot$ is $h$-heavy.
  \item [(ii)] Additionally, if $v$ is $h$-heavy, then every edge on the path from $v$ to leaf $u_h^\bot$ is a majority edge.
  \item [(iii)] Lastly, for any vertex $v$, there exists at most one hypothesis $h$ such that $v$ is \emph{$h$-heavy}.
  \end{enumerate}
\end{lemma}
\begin{proof}
  Item (i) is true by Lemma~\ref{Z:lem:alg-9}, which says that $p^-(v)$ decreases as one descends the tree.

  For (ii), it suffices to prove, by the first part, that for every $h$-heavy vertex $v$, the first edge on the path from $v$ to $h$ is a majority edge.
  Suppose for contradiction that there exists $h$ and an $h$-heavy vertex $v$ with a minority child $u$ such that $h$ is a descendant of $u$.
  Then $p^-(v) \ge p(u) \ge p_h$, which contradicts the definition of $v$ being $h$ heavy.

  For (iii), suppose for contradiction there exists two hypotheses $h$ and $h'$ such that $v$ is both $h$-heavy and $h'$-heavy.
  Since our \textsc{Decision~Tree} instance is well defined, there exists some test $j$ that distinguishes $h$ and $h'$, 
  i.e.\ $\tau_j(h)\neq \tau_j(h')$.
  As $v$ is $h$-heavy, we have $p_h > p^-(v)$. 
  Hence the answer for hypothesis $h$ under test $j$ is $k^+_{j,L(v)}$, the answer to $\tau_j$ accounting for the maximum weight of hypotheses in $L(v)$: if not choosing test $j$ at vertex $v$ would make the weight of hypotheses consistent with a minority child $v$ to be $\ge p_h$. 
  This is a contradiction as $p_h>p^-(v)$ and the tree $\mathcal{T}_G$ is greedy.
  However, as $v$ is also $h'$-heavy, we have, by the same reasoning, that $\tau_j(h')=k^+_{j,L(v)}$.
  This is a contradiction, as test $j$ was chosen to distinguish $h$ and $h'$.
\end{proof}
We now define some notation for dealing with non-uniform weights $p_h$, which are well-defined by Lemma~\ref{Z:lem:heavy-1}.
\begin{definition}
  For hypothesis $h\in[n]$, let $u_h^\top$ be the maximal ancestor of $h$ that is $h$-heavy.
  For vertex $v$, if there exists an $h$ such that $v$ is $h$-heavy, let $q(v)=p_h$, and otherwise let $q(v)=0$.
\end{definition}

%
\subsubsection{Defining Weighted Min Sum Set Cover}
Recall $I_{j,S}^+ = \tau_j^{-1}(k_{j,S}^+)$ and $I_{j,S}^- = [n] \setminus I_{j,S}^+$.
\begin{definition} 
Let $\WMSSC\ind{P}$ denote the instance of \emph{weighted min sum set cover} that is \emph{induced} by the chain $P=(P_1,\dots,P_{|P|})$. 
This instance is given by
\begin{itemize}
\item universe $S\defeq L(v_1)$ with weights $(p_h)_{h\in S}$,
\item for $j=1,\dots,m$, sets $A_j \defeq I_{j,S}^-\cap S$, and
\item for each $h=1,\dots,n$, a set $A_{m+h} \defeq \{h\}\cap S$ consisting of one element.\footnote{Some of these sets are empty, but we include them for notational convenience.}
\end{itemize}
Note we have a total of $m + n$ sets.
A \emph{solution} to the $\WMSSC$ problem is a permutation $\sigma:[m+n]\to [m+n]$ corresponding to an ordering of the sets, and the \emph{cost} $\WMSSC\ind{P}(\sigma)$ of a solution is the weighted sum of the \emph{cover times} of the elements in the universe $S$.
Formally,
\begin{align}
  \label{Z:eq:wmssc}
  \WMSSC\ind{P}(\sigma)
  \ \defeq \ \sum_{h\in S}^{} p_h\cdot \min\{\ell: h\in A_{\sigma(\ell)}\}
  \ = \ \sum_{\ell=1}^{m+n} p\left( S \setminus (A_{\sigma(1)}\cup\cdots\cup A_{\sigma(\ell-1)}) \right). 
\end{align}
Note that this instance is well defined, as each hypothesis $h\in L(v)$ is in some set $A_j$.
We sometimes refer to a solution $\sigma$ by the sets $A_{\sigma(1)},\dots,A_{\sigma(m+n)}$.
\end{definition}
\begin{remark}
  \label{Z:rem:mssc}
Since the initial \textsc{Decision~Tree} instance is well defined, any two elements can be distinguished by one of the $m$ tests.
Hence, there is at most one element $h\in L(v)$ such that, for all $j=1,\dots,m$, we have $h\notin A_j$. In other words, all but one of the sets $A_{m+h}$ for $h\in[n]$ are unnecessary.
\end{remark}
\begin{definition}
\label{def:greedy2}
We say a solution $\sigma:[m+n]\to[m+n]$ to $\WMSSC\ind{P}$ is \emph{greedy at index $\ell$} if the set $A_{\sigma(\ell)}$ covers the maximum number of elements not covered by sets $A_{\sigma(1)},\dots, A_{\sigma(\ell-1)}$.
We say a solution $\sigma:[m+n]\to[m+n]$ is \emph{greedy} if it is greedy at index $\ell$ for all $\ell\in[m+n]$,
\end{definition}
Note that, in the case of ties, there may be multiple greedy solutions to $\WMSSC\ind{P}$. 
Note also that, for any partial assignment $\sigma(1),\dots,\sigma(\ell)$, one can always complete the solution greedily, so that $\sigma:[m+n]\to[m+n]$ is greedy at indices $\ell+1,\ell+2,\dots,m+n$.
Definition~\ref{def:greedy2} lets us leverage the following theorem, due to Golovin and Krause, which generalizes Theorem~\ref{thm:flt}.\footnote{In fact, \citep{golovin2011adaptive} considers an even more general problems called \emph{Adaptive Stochastic Min-Sum Cover}.}
\begin{theorem}[Theorem 5.10 of \citep{golovin2011adaptive}]
  \label{Z:thm:flt-2}
  The greedy algorithm gives a 4-approximation to the WMSSC problem.
  Formally, let $\sigma$ be any greedy solution to $\WMSSC\ind{P}$, and let $\sigma_{\OPT}$ denote an optimal solution.
  We have
  \begin{align}
    \WMSSC\ind{P}(\sigma) \le 4\cdot \WMSSC\ind{P}(\sigma_{\OPT}).
  \end{align}
\end{theorem}

\subsubsection{Bounding chain weight above by WMSSC cost}
\begin{lemma}
\label{Z:lem:mssc}
  Let $s$ be a positive integer and let $P=(P_1,\dots,P_{|P|})$ be a level-$s$ chain.
  Then there exists a greedy solution $\sigma:[m+n]\to[m+n]$ to $\WMSSC\ind{P}$, such that
  \begin{align}
    \WMSSC\ind{P}(\sigma) \ge \sum_{\ell=1}^{|P|} \left(p(P_\ell) - q(P_\ell)\right).
  \end{align}
\end{lemma}
\begin{proof}
  Let $S$ be the universe of the instance $\WMSSC\ind{P}$, and let $A_1,\dots, A_{m+n}$ be the sets.
  For $\ell=1,\dots,|P|$, let $j_\ell$ be the test used at vertex $P_\ell$.
  Let $\ell_0\le|P|$ be the largest index such that $P_{\ell_0}$ is not $h$-heavy for any $h$, or 0 if no such index exists.
  If $\ell_0 < |P|$, let $h_0$ be the hypothesis such that $P_{\ell_0+1}$ is $h_0$-heavy.
  By Lemma~\ref{Z:lem:level-1}, all the edges along the path $P$ are majority edges.
  By Lemma~\ref{Z:lem:heavy-1}, for all $\ell_0 < \ell \le |P|$, vertex $P_\ell$ is $h_0$-heavy.
  Define a solution $\sigma:[m+n]\to[m+n]$ to $\WMSSC\ind{P}$ as follows.
  \begin{itemize}
  \item If $\ell_0 = |P|$, for $\ell=1,\dots,|P|$, let $\sigma(\ell)=j_\ell$ and complete the solution $\sigma$ greedily.
  \item Otherwise, for $1\le \ell\le \ell_0$, let $\sigma(\ell)=j_{\ell}$, let $\sigma(\ell_0+1)=m+h_0$, let $\sigma(\ell+1)=j_\ell$ for $\ell_0 < \ell \le |P|$, and complete the solution $\sigma$ greedily.
  \end{itemize}
  We claim $\sigma$ is a greedy solution.
  To prove this, we show the following. 
  \begin{itemize}
  \item[(i)] For all $j\in[m]$ and $\ell=1,\dots,|P|$, the majority answer for test $j$ with respect to vertex $P_1$ is the same as the majority answer for test $j$ with respect to vertex $P_\ell$.
  Equivalently, for all $j\in[m]$ and $\ell=1,\dots,|P|$, we have $A_j=I_{j,S}^-\cap S = I_{j,L(P_\ell)}^-\cap S$.
  As an immediate consequence, we know $A_{j_\ell}$ contains all the hypotheses in $L^-(P_\ell)$ and none of the hypotheses in $L(P_\ell^+)$.
  \item[(ii)] The set of hypotheses of $S$ not covered by $A_{j_1}, \dots,A_{j_{\ell-1}}$ is exactly $L(P_\ell)$.
  \item[(iii)] For each $1\le \ell\le |P|$, among sets $A_1,\dots,A_m$, set $A_{j_\ell}$ covers the maximum weight of hypotheses in $L(P_\ell)$, i.e.\ we have
  \begin{align}
    p(A_{j_{\ell}}\cap L(P_\ell))=\max_{j\in[m]}p(A_j\cap L(P_\ell)).
  \end{align}
  \item[(iv)] For each $1\le \ell \le\ell_0$, among sets $A_1,\dots,A_{m+n}$, set $A_{j_\ell}$ covers the maximum weight of hypotheses in $L(P_{\ell})$.
  \item[(v)] If $\ell_0<|P|$, then, among sets $A_1,\dots,A_{m+n}$, set $A_{m+h_0}=\{h_0\}$ covers the maximum weight of hypotheses in $L(P_{\ell_0+1})$.
  \item[(vi)] If $\ell_0<|P|$, then, for $\ell_0 < \ell \le|P|$, among sets $A_1,\dots,A_{m+n}$, set $A_{j_\ell}$ covers the maximum weight of hypotheses in $L(P_{\ell})\setminus \{h_0\}$.
  \end{itemize}
  These points suffices for proving that $\sigma$ is greedy.
  If $\ell_0=|P|$, items (ii) and (iv) tell us that $\sigma$ is greedy at indices $1,\dots,|P|$, so by construction $\sigma$ is greedy.
  If $\ell_0<|P|$, then (iv), (v), and (vi) tell us that $\sigma$ is greedy at indices $1,\dots,|P|+1$, so $\sigma$ is greedy.

  To show (i), fix $j\in[m]$ and $\ell\in\{1,\dots,|P|\}$.
  As $P_\ell$ is level-$s$ imbalanced, we also have $p(P_\ell)>2\delta^s$ and $p^-(P_\ell) \le \delta^s$ and $p(P_\ell^+) > \delta^s$, so $k_{j,L(P_\ell)}^+$ is the unique answer in $[K]$ accounting for more than half of the weight of hypotheses in $L(P_\ell)$.
  On the other hand, as vertex $P_1$ is level-$s$ imbalanced, we have $p(I_{j,S}^-\cap L(P_\ell)) \le p(I_{j,S}^-\cap S) \le p^-(P_1) \le \delta^s$, so the majority answer $k_{j,S}^+$ for test $j$ with respect to hypothesis set $S$ is exactly the answer described in the previous sentence. 
  Hence $k_{j,S}^+ = k_{j,L(P_\ell)}^+$.

  Item (ii) follows because $L(P_\ell)$ is the set of hypotheses consistent with $P_\ell$, which was obtained by following the majority edges from $P_1$.
  This means $L(P_\ell)$ contains all the hypotheses of $S$ not a consistent with a minority child of one of $P_1,\dots,P_{\ell-1}$.
  By the last paragraph, this is exactly $S\setminus (A_{j_1}\cup\cdots\cup A_{j_{\ell-1}})$.

  For (iii), at vertex $P_\ell$ in the greedy decision tree, the test index $j=j_\ell$ maximizes the weight $p(I_{j,L(P_\ell)}^-\cap L(P_\ell))$.
  By (i), this index $j$ equivalently maximizes $p(A_j\cap L(P_\ell))$, as desired.

  For (iv), at step $\ell$ for $\ell\le \ell_0$, by (i), the set $A_{j_\ell}$ covers a $p^-(P_\ell)$ weight of hypotheses in $L(P_\ell)$, which is more than $p_{h_0}$ by definition of $\ell_0$. 
  By (iii), $A_{j_\ell}$ covers at least as much weight of hypotheses in $L(P_\ell)$ as any of $A_1,\dots,A_m$, and, by Remark~\ref{Z:rem:mssc} and the previous sentence, at least as much as any of $A_1,\dots,A_{m+n}$.
  
  For (v), by maximality of $\ell_0$, we have $p_{h_0} > p^-(P_{\ell_0+1})$.
  Hence, the singleton $\{h_0\}$ covers more weight of hypotheses in $L(P_{\ell_0+1})$ than any of $A_1,\dots,A_m$, and thus, by Remark~\ref{Z:rem:mssc}, than any of $A_1,\dots,A_{m+n}$.
  
  For (vi), if there exists $h_0$ such that some $P_\ell$ is $h_0$-heavy, then, for any $j=1,\dots,m$, we have $h_0\notin A_j$ and $A_j\cap L(P_\ell) = A_j\cap (L(P_\ell)\setminus\{h_0\})$.
  Hence the $A_j$ among $A_1,\dots,A_m$ that covers the most weight of $L(P_\ell)\setminus \{h_0\}$ is $A_{j_\ell}$ by (iv).
  By Remark~\ref{Z:rem:mssc}, the only set among $A_{m+1},\dots,A_{m+n}$ that could cover a larger weight of $L(P_\ell)\setminus\{h_0\}$ is $\{h_0\}$, but it in fact covers 0 weight of $L(P_\ell)\setminus\{h_0\}$, so among sets $A_1,\dots,A_{m+n}$, set $A_{j_\ell}$ covers the maximum weight of hypotheses in $L(P_\ell)\setminus\{h_0\}$.
  This completes the proof that $\sigma$ is greedy.

  We now return to the proof of Lemma~\ref{Z:lem:mssc}.
  Take the greedy solution $\sigma$ given above.
  If $\ell_0=|P|$, then the set of vertices of $S$ not covered by $A_{\sigma(1)},\dots,A_{\sigma(\ell-1)}$ is exactly $L(P_\ell)$, which has weight $p(P_{\ell})$.
  Hence, by \eqref{eq:wmssc},
  \begin{align}
    \WMSSC\ind{P}(\sigma)
    \ &\ge \ \sum_{\ell=1}^{|P|} p(P_\ell)
    \ = \ \sum_{\ell=1}^{|P|} (p(P_\ell) - q(v)). 
  \end{align}
  Now suppose $\ell_0<|P|$.
  Recall that the definition of $\ell_0$ implies $P_{\ell_0+1},P_{\ell+2},\dots,P_{|P|}$ are all $h_0$-heavy.
  Hence, we have $q(P_{\ell}) = 0$ for $1\le \ell\le \ell_0$, and $q(P_{\ell})=p_{h_0}$ for $\ell_0 < \ell \le |P|$.
  Thus,
  \begin{align}
    \WMSSC\ind{P}(\sigma)
    \ &\ge \ \sum_{\ell=1}^{|P|+1}  p\left( S\setminus(A_{\sigma(1)}\cup\cdots\cup A_{\sigma(\ell-1)}) \right) \nonumber\\ 
    \ &= \ \sum_{\ell=1}^{\ell_0+1} p(P_\ell) + \sum_{\ell=\ell_0+1}^{|P|} p(L(P_\ell)\setminus\{h_0\}) \nonumber\\
    \ &> \ \sum_{\ell=1}^{\ell_0} p(P_\ell) + \sum_{\ell=\ell_0+1}^{|P|} p(L(P_\ell)\setminus\{h_0\}) \nonumber\\
    \ &= \ \sum_{\ell=1}^{|P|} \left(p(P_\ell) - q(P_\ell)\right)
  \end{align}
  as desired.
  In the last equality, we used that $h_0\in L(P_\ell)$ for $\ell=1,\dots,|P|$.
\end{proof}
Let $\sigma_G\ind{P}:[m+n]\to[m+n]$ be the greedy solution to $\WMSSC\ind{P}$ given by Lemma~\ref{Z:lem:mssc}, and let $\sigma_{\OPT}\ind{P}$ be an optimal solution to $\WMSSC\ind{P}$.

\subsubsection{Bounding WMSSC cost above by $C_{\OPT}$}
\begin{lemma}
  \label{Z:lem:sum-2}
  Let $s$ be a positive integer.
  We have
  \begin{align}
    C_{\OPT}\ge -1 + \sum_{P\in \mathcal{P}_s}^{} \WMSSC\ind{P}(\sigma_{\OPT}\ind{P}).
  \end{align}
\end{lemma}
\begin{proof}
  Let $S\ind{P}$ be the universe of the instance $\WMSSC\ind{P}$, and let $A_1,\dots, A_{m+n}$ be the sets.
  Construct a path $w_1,\dots,w_{\ell\sar}$ in $\mathcal{T}_{\OPT}$ such that $w_1$ is the root, $w_{\ell\sar}$ is a leaf for hypothesis $h\sar$, and, for $\ell=1,\dots,\ell\sar-1$, if the test at vertex $w_\ell$ has index $j_\ell$, the edge to its child $w_{\ell+1}$ corresponds to the answer $k_{j_\ell,S\ind{P}}^+$, the majority answer of test $j_\ell$ with respect to set $S\ind{P}$.
  Suppose the test at vertex $w_\ell$ in the optimal tree has index $j_\ell$.
  Since we follow the edges with label $k^+_{j,S\ind{P}}$, this corresponds to following the path for an hypothesis contained in $I_{j_\ell,S\ind{P}}^+$. 
  In other words, we have, for $\ell=1,\dots,\ell\sar-1$,
  \begin{align}
    L(w_{\ell+1}) 
    \ = \ L(w_\ell)\setminus I_{j_{\ell},S\ind{P}}^-
    \ = \ L(w_\ell)\setminus A_{j_{\ell}}. 
  \end{align}
  Thus the sequence $A_{j_1},\dots,A_{j_{\ell\sar}}, \{h\sar\}$ covers $[n]$, and hence $S\ind{P}$, and thus gives a valid solution $\sigma_{TREE}$ to the instance $\WMSSC\ind{P}$, where $\sigma_{TREE}(1) = j_1, \sigma_{TREE}(2)=j_2,\dots,\sigma_{TREE}(\ell\sar)=j_{\ell\sar}, \sigma_{TREE}(\ell\sar+1)=m+h\sar$, and $\sigma_{TREE}$ on larger indices is arbitrarily chosen.
  Note that the depth of a hypothesis $h$ in the tree is at least the number of vertices of $w_1,\dots,w_{\ell\sar}$ that are on the root-to-leaf path of $h$, and this number is $\min\{\ell:h\in A_{j_\ell}\}$, except for $h^*$, in which case it is 1 smaller.
  Then we have
  \begin{align}
    \sum_{h\in S\ind{P}}^{} p_h\cdot d_{\OPT}(h) 
    \ \ge  \ -p_{h\sar} +\sum_{h\in S\ind{P}}^{} p_h\cdot \min\{\ell:h\in A_{j_\ell}\} 
    \ &\ge \ \WMSSC\ind{P}(\sigma_{TREE}) - p_{h\sar}  \nonumber\\
    \ &\ge \ \WMSSC\ind{P}(\sigma\ind{P}_{\OPT}) - \max_{h\in S\ind{P}}p_h.
  \end{align}
  Summing over $P\in\mathcal{P}_s$ gives
  \begin{align}
    C_{\OPT}
    \ \ge \ \sum_{P\in\mathcal{P}_s}^{} \left( \WMSSC\ind{P}(\sigma\ind{P}_{\OPT}) - \max_{h\in S\ind{P}}p_h \right) \ge -1 + \sum_{P\in\mathcal{P}_s}^{} \left( \WMSSC\ind{P}(\sigma\ind{P}_{\OPT}) \right). 
  \end{align}
\end{proof}

\begin{lemma}
\label{Z:lem:sum-3}
  We have 
  \begin{align}
    \sum_{v\text{ imbalanced}}^{} p(v)
    \ \le \ 4s_{\max}(C_{\OPT}+1) + \sum_{v\in\mathcal{T}_G^{\mathrm{o}}}^{} q(v).
  \end{align}
\end{lemma}
\begin{proof}
  Each imbalanced vertex $v$ is level-$s$ imbalanced for some positive integer $s$, so it is part of some level-$s$ chain, $P$.
  Note that $p(v) - q(v) \ge 0$ for all vertices $v$.
  Hence,
  \begin{align}
    \label{Z:eq:sum-5}
    \sum_{v\text{ imbalanced}}^{} p(v) - \sum_{v\in\mathcal{T}_G^{\mathrm{o}}}^{} q(v)
    \ &\le \  \sum_{v\text{ imbalanced}}^{} (p(v) - q(v)) \nonumber\\
    \ &\le \ \sum_{s=1}^{s_{\max}} \sum_{\substack{P\in\mathcal{P}_s \\ P=(P_1,\dots,P_{|P|})}}^{} \sum_{\ell=1}^{|P|} \left(p(P_\ell) - q(P_\ell)\right) \nonumber\\
    \ &\le \ \sum_{s=1}^{s_{\max}} \sum_{\substack{P\in\mathcal{P}_s \\ P=(P_1,\dots,P_{|P|})}}^{} \WMSSC\ind{P}( \sigma_G\ind{P} ) \nonumber\\
    \ &\le \ \sum_{s=1}^{s_{\max}} \sum_{P\in\mathcal{P}_s}^{} 4\WMSSC\ind{P}(\sigma\ind{P}_{\OPT}) \nonumber\\ 
    \ &\le \  4s_{\max}(C_{\OPT}+1) 
  \end{align}
  The first inequality is because $q(v)\ge 0$ for all $v$.
  The second inequality is because $p(v)-q(v) \ge 0$ for all $v$, and that every imbalanced $v$ is in some chain.
  The third inequality is by Lemma~\ref{Z:lem:mssc}.
  The fourth inequality is by Theorem~\ref{Z:thm:flt-2}.
  The fifth inequality is by Lemma~\ref{Z:lem:sum-2}.
  Rearranging gives the desired result.
\end{proof}

\subsection{Bounding the cost contribution of heavy vertices}

We bound the cost contribution of the heavy vertices via a connection to SET-COVER.
A theorem due to Lovasz \citep{lovasz1975ratio}, Johnson \citep{johnson1974approximation}, Chvatal \citep{chvatal1979greedy}, and Stein \citep{Stein74} states that, in any instance of SET-COVER where no set covers more than $h$ elements, the greedy algorithm gives a $1+\ln h$ approximation.
We show a generalization of this result, based on the following the definition.
\begin{definition}
  Let $\Phi$ be an instance of SET-COVER with a universe $S$ and sets $A_1,\dots,A_m$.
  Let $\bfp=(p_h)_{h\in S}$ be a sequence of weights assigned to the elements of $S$.
  A \emph{$\bfp$-weighted greedy algorithm} for SET-COVER is repeatedly chooses the set $A_j$ that minimizes $\sum_{h\in A_j\cap S'}^{} p_h$, where $S'$ is the set of uncovered elements.
\end{definition}
\begin{theorem}
  Let $\Phi$ be an instance of SET-COVER with a universe $S$ and sets $A_1,\dots,A_m$.
  Let $\bfp=(p_h)_{h\in S}$ be a sequence of weights assigned to the elements of $S$.
  Then, if the optimal solution to $\Phi$ has uses at most $\ell_{\OPT}$ sets, then the $\bfp$-weighted greedy algorithm uses at most $(1 + \ln\frac{p^*}{\min_h p_h})\ell_{\OPT}$ sets, where $p^*\defeq\max_{j\in[m]}\sum_{h\in A_j}^{} p_h$.
  \label{Z:thm:sc}
\end{theorem}
While this argument may be known, we are not aware of a known reference, so we provide a proof for completeness in Appendix~\ref{app:sc}.

We now bound $\sum_{v\in\mathcal{T}_G^{\mathrm{o}}}^{} q(v)$.
\begin{lemma}
  \label{Z:lem:alg-29}
  For all $h \in[n]$, we have
  \begin{align}
    d_{\mathcal{T}_G}(u_h^\bot) - d_{\mathcal{T}_G}(u_h^\top) \le \left(1 + \ln\left( \frac{p_{\max}}{p_{\min}} \right)\right)\cdot d_{\OPT}(h).
  \end{align}
\end{lemma}
\begin{proof}
  Fix $h$.
  Let $v_1=u_h^\top,\dots,v_{\ell_G}=u_h^\bot$ denote the path from vertex $u_h^\top$ to leaf $u_h^\bot$ in the greedy tree. 
  Let $\SC\ind{h}$ denote the \emph{SET-COVER} instance with the following parameters:
  \begin{itemize}
  \item Universe $S=L(u_h^\top)\setminus\{h\}$
  \item For $j=1,\dots,m$, sets $A_j=I_{j,S}^-\cap S$.
  \end{itemize}
  Let $\ell_{\OPT}$ denote the cost of the optimal solution to $\SC\ind{h}$, and for $\ell=1,\dots,\ell_G$, let $j_\ell$ denote the test at vertex $v_\ell$.

  We make the following observations.
  First, by definition of $u_h^\top$, for all $j\in[m]$, the set $I_{j,S}^-$ does not contain hypothesis $h$: vertex $u_h^\top$ satisfies $p^-(u_h^\top) < p_h$ so the answer $k^+_{j,S}$ of a test $j$ that accounts for the largest weight of hypotheses in $S$ always contains the hypothesis $h$, as any other answer, by the definition of the greedy algorithm, has weight at most $p^-(u_h^\top)$.
  
  Second, the sets $A_{j_1}, A_{j_2},\dots, A_{j_{\ell_G}}$ form a $\bfp$-weighted greedy solution for this SET-COVER instance $\SC\ind{h}$, where $\bfp=(p_h)_{h\in S}$.
  For $\ell \le \ell_G$, the first $\ell-1$ sets in the above sequence cover all of $S$ except the elements of $L(v_{\ell})$.
  In the greedy decision tree, the index $j=j_\ell$ that maximizes $p(I_{j,L(v_{\ell})}^-\cap L(v_{\ell}))$.
  Note that, for any $j$, the set $I_{j,L(v_{\ell})}^-$ are the hypotheses for the $K-1$ answers of $\tau_j$ that exclude hypothesis $h$.
  However, the set $I_{j,S}^-$ also contains exactly the hypotheses for the $K-1$ answers of $\tau_j$ that exclude hypothesis $h$.
  Hence $I_{j,L(v_{\ell})}^-\cap S = I_{j,S}^-\cap S = A_j$, so $j=j_\ell$ maximizes $p(A_j\cap L(v_\ell))$.
  Thus, sets $A_{j_1}, A_{j_2},\dots, A_{j_{\ell_G}}$ form a $\bfp$-weighted greedy solution for this SET-COVER instance $\SC\ind{h}$.
  Hence, we may apply Theorem~\ref{Z:thm:sc}.
  Since $\max_{j\in[m]}p(A_j)=p^-(u_h^\top) < p_h\le p_{\max}$ and $\min_{h\in S}p_h \ge p_{\min}$, we have, by Theorem~\ref{Z:thm:sc},
  \begin{align}
    d_{\mathcal{T}_G}(u_h^\bot) - d_{\mathcal{T}_G}(u_h^\top)
    \ = \  \ell_G-1
    \ \le \ \left(1+\ln\frac{p_{\max}}{p_{\min}}\right)\ell_{\OPT}.
  \end{align}
  
  It remains to prove $\ell_{\OPT}\le d_{\OPT}(h)$.
  Let $w_1,\dots,w_{\ell_T}=h$ be the root-to-leaf path for hypothesis $h$ in the optimal tree $\mathcal{T}_{\OPT}$.
  For $\ell=1,\dots,\ell_T-1$, set $j_\ell'$ to be the test chosen at vertex $w_\ell$ in the optimal tree.
  By the first point, the set $I_{j,S}^-$ contains the hypotheses for the $K-1$ answers of test $\tau_j$ that exclude hypothesis $h$.
  As the optimal tree is a decision tree, any test is distinguished from $h$ by one of tests $j_1,\dots,j_{\ell_T-1}$, so $A_{j_1'}, A_{j_2'},\dots,A_{j_{\ell_T-1}'}$ cover all of $[n]\setminus\{h\}$, and thus covers $S$. 
  Furthermore, $\ell_T-1=d_{\OPT}(h)$, so there is a solution to $\SC\ind{h}$ of size $d_{\OPT}(h)$.
  Hence, $d_{\OPT}(h)\ge \ell_{\OPT}$, as desired.
\end{proof}
As a corollary, we have
\begin{lemma}
\label{Z:lem:sum-4}
  \begin{align}
    \sum_{v\in\mathcal{T}_G^\mathrm{o}}^{} q(v) 
    \ &\le \ \left( 1 + \ln \frac{p_{\max}}{p_{\min}} \right)C_{\OPT}. 
  \end{align}
\end{lemma}
\begin{proof}
  We have
  \begin{align}
    \sum_{v\in\mathcal{T}_G^\mathrm{o}}^{} q(v)
    \ &= \  \sum_{h\in[n]}^{} p_h\cdot (d_{\mathcal{T}_G}(u_h^\bot) - d_{\mathcal{T}_G}(u_h^\top)) \nonumber\\
    \ &\le \ \sum_{h\in[n]}^{} \left( 1 + \ln\frac{p_{\max}}{p_{\min}} \right)\cdot d_{\OPT}(h) \nonumber\\
    \ &= \ \left( 1 + \ln\frac{p_{\max}}{p_{\min}} \right)\cdot C_{\OPT}.
  \end{align}
  The first equality uses Lemma~\ref{Z:lem:heavy-1}, which tells us that every vertex between $u_h^\top$ and $h$ is $h$-heavy, and no such vertex is $h'$-heavy for $h'\neq h$.
  Furthermore, these are the only $h$-heavy vertices as $u_h^\top$ is the maximal $h$-heavy vertex.
  The inequality is by Lemma~\ref{Z:lem:alg-29}.
\end{proof}

\subsection{Finishing the proof}
\begin{proof}[Proof of Theorem~\ref{thm:alg-main}]
  We have
  \begin{align}
    C_G
    \ &= \  \sum_{v\in\mathcal{T}_G^\mathrm{o}}^{} p(v) \nonumber\\
    \ &= \  \sum_{v\text{ balanced}}^{} p(v) + \sum_{v\text{ imbalanced}}^{} p(v) \nonumber\\
    \ &\le \  \frac{\log n}{\log \frac{2}{\delta}}\cdot \frac{2}{\delta} + \sum_{v\text{ imbalanced}}^{} p(v)  
        & (\text{Lemma~\ref{Z:lem:level-3}})\nonumber\\
    \ &\le \  \frac{\log n}{\log 2C_{\OPT}}\cdot 2C_{\OPT} + 4s_{\max}(C_{\OPT}+1) + \sum_{v\in\mathcal{T}_G^\mathrm{o}}^{} q(v)
        & (\text{Lemma~\ref{Z:lem:sum-3}})\nonumber\\
    \ &\le \  \frac{\log n}{\log 2C_{\OPT}}\cdot 2C_{\OPT} + 4\cdot \frac{\log\frac{1}{p_{\min}}}{\log \frac{1}{\delta}} \cdot (C_{\OPT}+1) + \left( 1 + \ln\frac{p_{\max}}{p_{\min}} \right)\cdot C_{\OPT}
        & (\text{Lemma~\ref{Z:lem:sum-4}})\nonumber\\
    \ &< \   \left(\frac{12\log\frac{1}{p_{\min}}}{\log C_{\OPT}} + \ln\frac{p_{\max}}{p_{\min}}\right)\cdot C_{\OPT},
  \end{align}
  as desired.
  In the last inequality, we used that (i) $\log 2C_{\OPT} > \log C_{\OPT}$, (ii) $\log n \le \log \frac{1}{p_{\min}}$, (iii) $1+C_{\OPT} < 2C_{\OPT}$, and $1 \le \frac{\log \frac{1}{p_{\min}}}{\log C_{\OPT}}$.
\end{proof}

\section{Proof of Theorem~\ref{Z:thm:sc}}
\label{app:sc}
We closely follow the argument of Chvatal~\citep{chvatal1979greedy}.
Suppose the $\vec p$-weighted greedy algorithm uses $\ell_G$ sets.
By re-indexing the sets, we may assume without loss of generality that the $\bfp$-weighted greedy algorithm chooses sets $A_1,\dots,A_{\ell_G}$ in that order.
For $r=1,\dots,\ell_G$ and $j=1,\dots,m$, let $A_j\ind{r}$ denote the elements of set $A_j$ not covered by the first $r-1$ chosen sets.
For $r=1,\dots,\ell_G$ and $j=1,\dots,m$, let $\rho_j\ind{r} = \sum_{h\in A_j\ind{r}}^{} p_h$ denote the sum of the weights of the elements of $A_j\ind{r}$.
For $h=1,\dots,n$, let $y_h=\frac{1}{\rho_r\ind{r}}$, where $r$ is the index at which element $h$ is first covered.
Equivalently, $r$ is the unique index such that $h\in A_r\ind{r}$.
In this way, we have
\begin{align}
  \sum_{h\in [n]}^{} p_h y_h = \sum_{r=1}^{\ell_G} \rho_r\ind{r}\cdot\frac{1}{\rho_r\ind{r}} = \ell_G,
\label{eq:sc-1}
\end{align}
and, for all $j=1,\dots,m$,
\begin{align}
  \sum_{h\in A_j}^{} p_hy_h 
  \ &= \  \sum_{r=1}^{\ell_j} \frac{\rho_j\ind{r}-\rho_j\ind{r+1}}{\rho_r\ind{r}},
\label{eq:sc-2}
\end{align}
where $\ell_j$ is the largest index such that $\rho_j\ind{\ell_j} > 0$.
Hence using that $\rho_j\ind{1},\rho_j\ind{2},\dots,\rho_j\ind{\ell_j}$ is a non-increasing sequence, we have
\begin{align}
  \sum_{h\in A_j}^{} p_hy_h 
  \ &= \  \sum_{r=1}^{\ell_j} \frac{\rho_j\ind{r}-\rho_j\ind{r+1}}{\rho_j\ind{r}}
  \ = \ 1 + \sum_{r=1}^{\ell_j-1}  \frac{\rho_j\ind{r}-\rho_j\ind{r+1}}{\rho_j\ind{r}} \nonumber\\
  \ &\le \ 1 + \sum_{r=1}^{\ell_j-1} \int_{\rho_j\ind{r+1}}^{\rho_j\ind{r}} \frac{dx}{x} 
  \ = \ 1+\int_{\rho_j\ind{\ell_j}}^{\rho_j\ind{1}} \frac{dx}{x} 
  \ = \ 1+\ln\frac{\rho_j\ind{1}}{\rho_j\ind{\ell_j}}
  \ \le \ 1+\ln\frac{p_{\max}}{p_{\min}}. 
\label{eq:sc-3}
\end{align}
Let $J\subset [m]$ denote the indices of the optimal cover for $\Phi$.
Applying \eqref{eq:sc-2} and summing \eqref{eq:sc-3} for $j\in J$, we have
\begin{align}
  \ell_G 
  \ &= \ \sum_{h\in [n]}^{} p_hy_h
  \ \le \  \sum_{j\in J}^{} \sum_{h\in A_j}^{} p_hy_h
  \ \le \ |J|\cdot \left( 1+\ln\frac{p_{\max}}{p_{\min}} \right)
  \ = \ \ell_{\OPT}\cdot  \left( 1+\ln\frac{p_{\max}}{p_{\min}} \right)
\label{}
\end{align}
as desired.

\section{Tightness of Theorem~\ref{thm:alg-main}}
\label{app:tight}
In this section, we prove Propositions~\ref{thm:avg-lower} and \ref{thm:lower_bound_poly_ratio}, which show two ways that Theorem~\ref{thm:alg-main} is tight.
\subsection{Proof of Proposition~\ref{thm:avg-lower}}
\label{app:B}
\begin{proof}
We prove Proposition~\ref{thm:avg-lower} with the stronger guarantee that $C_\OPT \le 4C^*$ when $C^*$ is an integer.
Then, taking $(C^*)' = \ceil{C^*} \le 2C^*$ gives the desired result.

When $n$ is sufficiently large, for $C^*\ge n^{1/4}$, the statement is trivial, as any instance for which $C_{OPT}\in[\frac{1}{4} C^*,4C^*]$, satisfies the requirements.
Number the hypotheses $1,\dots,n^*$.
Let $n^*\in(n-C^*,n]$ be such that $C^*|n^*$. 
Place the hypotheses $1,\dots,n^*$ in a grid with $C^*$ columns and $r\defeq \frac{n^*}{C^*}$ rows, numbered $1,\dots,r$, so that each grid square contains at most 1 hypothesis.
Recursively identify a family of \emph{good} sets of rows as follows: $[r]$ is good, and for every good set $A'$ containing $r'>1$ rows, create a partition $A'=A'_1\cup A'_2$ such that $|A'_1| = 1+\floor{r'/C^*}$ and $|A'_2| = r' - |A_1'|$, and identify $A'_1$ and $A'_2$ as good.
Define three types of tests: 
\begin{enumerate}
\item For each of $h=n^*+1,\dots,n$, a test that outputs 1 if the hypothesis is $h$ and 2 otherwise.
\item For each column $c$, a tests that outputs 1 if the hypothesis is in column $c$ and 2 otherwise.
\item For each column $c$ and $t\in[0,\log r]$, a test that outputs 1 if the hypothesis is in column $c$ and the $t$th digit of the row number's binary expansion is a one, and 2 otherwise.
\item for each good set $A'$, a test that outputs 1 if the row of $h$ is in $A'$ and 2 otherwise.
\end{enumerate}

Let $h$ be the unknown hypothesis.
There is a strategy that first checks whether $h$ is one of $n^*+1,\dots,n$, for a total of at most $n-n^* < C^*$ queries.
If not, the strategy identifies the column containing $h$ in at most $C^*$ queries using tests of type 2 and then identifies the corresponding row using tests of type 3, which takes $\log r < \log n\le 2C^*$ queries. 
We thus need at most $4C^*$ queries for each hypothesis, so here
\begin{align}
  C_{OPT}\le 4C^*.
\end{align}

The greedy strategy uses tests of type 4, trying to first find the row containing $h$.
This is because, for tests 1, 2, and 3, one answer accounts for at least $1-\frac{1}{C^*}$ fraction of the remaining hypotheses (it accounts at least $1-\frac{1}{C^*}$ fraction of the columns in the grid), and if the candidate set of rows containing $h$ is a good set $A'$, under the membership test for the good set $A_1'$, all answers account for less than $1-\frac{1}{C^*}$ fraction of the remaining hypotheses.

When $h$ is chosen uniformly at random from $1,\dots,n^*$, the row containing $h$ is a uniformly random row.
While there are at least $C^*$ candidate rows containing $h$, each test gives at most $\mathbb{H}(2/C^*)$ bits of information about the row containing $h$ in expectation (over the randomness of $h$).
Since the row containing $h$ has at least $\log r  = \log (n^*/C^*) > \frac{1}{2}\log n$ bits of information, and the row has at most $\log C^*$ bits of information when there are at most $C^*$ candidate rows remaining, we have, by an analysis similar to Lemma~\ref{lem:type-3}, the greedy algorithm takes at least $\frac{\frac{1}{2}\log n - \log C^*}{\mathbb{H}(2/C^*)}$ queries to identify the row containing $h$ on average.
Hence,
\begin{align}
  C_G 
  \ge \frac{n^*}{n}\cdot \frac{\frac{1}{2}\log n - \log C^*}{\mathbb{H}(2/C^*)} 
  \ge \frac{1}{2}\cdot \frac{\frac{1}{4}\log n}{\mathbb{H}(2/C^*)} 
  \ge \frac{C^*\log n}{16\log C^*} 
\end{align}
as desired.
In the last inequality, we used that $\mathbb{H}(2x) \le -4x\log(x)$ for $x\in[0,1]$.
\end{proof}

\subsection{Proof of Proposition~\ref{thm:lower_bound_poly_ratio}}
\label{app:E}

In this appendix, we use that it is NP-hard to approximate \textsc{Set~Cover} to within a factor of $\frac{1}{2} \log n_0$ \cite{moshkovitz2012projection}.
\begin{theorem}
Let $r\in(0,1)$. 
Then, for $n$ sufficiently large, approximating $\DT(2 n^r \log n)$ to a factor of $\frac{1}{12} \log(n^r)$ is NP-hard.
\end{theorem}
\begin{proof}
We design a reduction from \textsc{Set~Cover} to $\DT(2n^r \log n)$.
Suppose we are given a \textsc{Set~Cover} instance with $n_0$ elements, $M$ sets $\{S_i\}$ where $S_i \subseteq [n_0]$, and an optimal cover of size $B_{OPT}$.
In polynomial time, we construct a $\DT(2n^r\log n)$ instance on $n\le n_0^{1/r}$ hypotheses such that, if $C_{OPT}$ is the optimal decision tree cost, then, for some $q$,
\begin{align}
\frac{1}{2} qB_{OPT}< C_{OPT}< 3q B_{OPT}.
\end{align}
The theorem follows as a $\frac{1}{2}\log n_0$ approximation to \textsc{Set~Cover} is NP-hard, and here $n^r\le n_0$.

Let $q = \floor{\frac{1}{r}\log n_0}$. 
Let $\ell = \floor{n_0^{1/r} / (n_0 q + 1)}$.
Identify the hypotheses by elements of $[\ell] \times ([q] \times [n_0] \cup \{\perp\})$.
In this way, there are $n \sim n_0^{1/r}$ hypotheses.
Let the elements of $[\ell]\times \{\perp\}$ have weight $\frac{1}{2n} + \frac{1}{2\ell}$, and let $p_h=\frac{1}{2n}$ for all other hypotheses $h\in[\ell]\times [q]\times [n_0]$.
In this way, for $n$ sufficiently large, we have $p_{max}/p_{min} = 1 + \frac{n}{\ell} \le 2n^r \log n$.
Create tests of the following forms:
\begin{enumerate}
    \item For each $i_1 \in[\ell], i_2 \in [q], j \in [M]$, define a test that outputs $1$ on hypotheses $(h_1,h_2,h_3)$ if and only if $i_1=h_1,i_2=h_2,h_3 \in S_{j}$.
    \item For each $i_1 \in[\ell], i_2 \in [q], j \in [M]$, and $t\in[0,\log n_0]$, define a binary test that outputs 1 if and only if $i_1=h_1,i_2=h_2,h_3\in S_{j}$, and the $t$th bit of $h_3$'s binary representation is 1.
    \item For each $t\in[0, \log \ell]$, define a binary test on $(h_1,*)$ that outputs 1 if and only if the $t$th bit of $h_1$'s binary representation is $1$.
\end{enumerate}

Consider a cover of $[n_0]$ using $B_{\text{OPT}}$ of the $M$ sets.
We can define a tree that, given a hypothesis $h=(h_1,*)$, first determines $h_1$ using tests of type 3.
Then, in each subtree, the set of consistent hypotheses is exactly $\{h_3\}\times ([q]\times [n_0]\cup\{\perp\})$.
In each subtree, one can isolate the $\perp$ hypothesis in $qB_{OPT}$ queries, using tests of type 1, and use tests of type 2 to identify the remaining hypotheses in $1+\log n_0$ tests each.
In each subtree, all of the hypotheses of the form $(h_1,\perp)$ have weight at most $\frac{1}{\ell}$ and depth at most $qB_{OPT}$, and all other hypotheses have weight $\frac{1}{2n}$ and depth at most $qB_{OPT} + 1+\log n_0$, so their total contribution to the cost of the tree is at most $qB_{OPT} + 1 + \log n_0$.
Assuming $n_0$ is sufficiently large,
\begin{align}
C_{OPT}
\ \leq \ 1+\log \ell + 2q B_{OPT} + \log n_0
\ < \ 3q B_{OPT}.
\end{align}

Now suppose we are given a solution to the $\DT(2n^r\log n)$ with cost $C_{OPT}$.
In the optimal tree, for all hypotheses $(h_1,\perp)$, at least $qB_{OPT}$ tests of type-1 or type-2 must appear on the root-to-leaf path of $(h_1,\perp)$: if not, there exists $h_2\in[q]$, such that at most $B_{OPT}-1$ tests of type-1 with parameters $(h_1,h_2,j)$ or type-2 with parameters $(h_1,h_2,j,t)$ were used. By taking the indices $j$ used in these tests, there are at most $B_{OPT}-1$ sets covering $[n_0]$, which is a contradiction.
Thus, each hypothesis $(h_1,\perp)$ has depth at least $qB_{OPT}$.
Since the hypotheses $(h_1,\perp)$ account for at least half of the weight of the hypotheses, we have
\begin{align}
  C_{OPT}\ge \frac{1}{2}qB_{OPT}.
\end{align}
This completes the proof.
\end{proof}

\section{Rounding weights}
\label{app:H}

\begin{proposition}
Suppose a \textsc{Decision~Tree} instance has weights $p_1,\dots,p_n$ and a cost function $C(\cdot)$. Then, there exist weights $p_1',\dots,p_n'$ such that $\min_i p_i' \geq \frac{1}{n(n-1)}$ and the cost function $C'(\cdot)$ of the associated instance satisfies $|C'(\mathcal{T}) - C(\mathcal{T})| \le 1$ for all decision trees $\mathcal{T}$.
\end{proposition}
\begin{proof}
Let $w_i'=\max(p_i, \frac{1}{(n-1)^2})$ and define $W' = \sum_i w_i'$. 
We have $W' \geq 1$ and $W' \leq 1 + \frac{n-1}{(n-1)^2} = \frac{n}{n-1}$.
Let $p_i' = \frac{w_i'}{W'}$, so that $p_i' \ge \frac{1}{n(n-1)}$ for all $i$.
Hence, for all decision trees $\mathcal{T}$,
\begin{align}
C'(\mathcal{T}) \ &= \   \frac{1}{W'} \sum_i w_i' d_\mathcal{T}(i) \geq \frac{1}{W'} \sum_i p_i d_\mathcal{T}(i)  \geq  \frac{n-1}{n} C(\mathcal{T}) \ge C(\mathcal{T})-1 \nonumber\\
C'(\mathcal{T}) \ &= \   \frac{1}{W'} \sum_i w_i' d_\mathcal{T}(i) \leq  \sum_i \left(p_i + \frac{1}{n^2}\right) d_\mathcal{T}(i)  \leq C(\mathcal{T}) + \sum_i \frac{1}{n^2} n \leq C(\mathcal{T}) + 1.\qedhere
\end{align}
\end{proof}

\end{document}